\documentclass{article}
\usepackage{amssymb}
\usepackage[numbers]{natbib}
\usepackage{enumitem}
\usepackage{pdfpages}
\usepackage{amssymb}
\usepackage{float}
\usepackage{pifont}
\usepackage{graphicx}
\usepackage{epstopdf}
\usepackage{url}
\usepackage{amsmath}
\usepackage{amsthm}
\usepackage[utf8]{inputenc}
\usepackage[T1]{fontenc}
\usepackage{lmodern}
\usepackage{ae}
\usepackage[normalem]{ulem}

\newtheorem{theorem}{\bf Theorem}[section]
 \newtheorem{proposit}[theorem]{\bf Proposition}
 \newtheorem{coro}[theorem]{\bf Corollary}
\newtheorem{lem}[theorem]{\bf Lemma}

 \newtheorem{example}[theorem]{\bf Example}
\newtheorem{fait}{\bf Claim}

\def\thm#1\par{\medskip\par\noindent\begin{theorem} \strut \sl #1 \end{theorem}\par}
\def\propo#1\par{\medskip\par\noindent\begin{proposit} \strut \sl #1 \end{proposit}
\par}
\def\cor#1\par{\medskip\par\noindent\begin{coro} \strut \sl #1 \end{coro}\par}
\def\lm#1\par{\medskip\par\noindent\begin{lem} \strut \sl #1 \end{lem}\par}
\def\fct#1\par{\medskip\par\noindent\begin{fait} \strut \sl #1 \end{fait}\par}
\date{}

\title{Topologies for Error-Detecting Variable-Length Codes}

\author{Jean N\'eraud\\
{\small{\it Univ Rouen Normandie, LITIS UR 4108, F-76000 Rouen, France}}}

\begin{document}
\maketitle
\begin{abstract}
Given  a finite alphabet $A$ and a quasi-metric $d$ over $A^*$, we introduce the relation $\tau_{d,k}\subseteq A^*\times A^*$ such  that $(x,y)\in\tau_{d,k}$ holds whenever  
$d(x,y)\le k$.
The error detection capability of variable-length codes  is  expressed in term of   conditions over  $\tau_{d,k}$. 
With respect to   the prefix metric, the  factor one, and any quasi-metric  associated with some free monoid (anti-)automorphism, we prove that one can decide whether a  given regular variable-length code  satisfies any of those error detection constraints.
\end{abstract}
\section{Introduction}
\label{Intro}

In Computer Science, the transmission of finite sequences of symbols (the so-called {\it words}) via some  channel
constitutes one of the most challenging research fields.
With the notation of the free monoid, some classical models 
may be informally described as indicated in the following:

Two finite  {\it alphabets}, say  $A$ and $B$,  are required,
every information being modeled by a unique word, say  $u$,  in $B^*$ (the {\it free monoid} generated by $B$). 
Usually, in order to  facilitate the transmission, beforehand $u$ is transformed in  $w\in A^*$, the so-called {\it input word}: this is  done by applying some fixed one-to-one {\it coding} mapping  $\phi: B^*\longrightarrow A^*$. 
In numerous cases, $\phi$ is an injective {\it monoid homomorphism},
whence $X=\phi(B)$ is a  {\it  variable-length code} (for short, a {\it code}): equivalently 
every equation among the words of $X$ is necessarily trivial. Such a translation  is particularly illustrated by the well-known examples of the Morse and Huffman codes. 
Next, $w$ is transmitted via a fixed {\it channel} into $w'\in A^*$, the so-called {\it output word}: should $w'$ be altered by some
{\it noise}
and then the resulting word $\phi^{-1}(w')\in B^*$ could be different from the initial word $u$.
In the most general model of  transmission, the channel is represented by some {\it probabilistic transducer}.
However, in the framework of error detection, most of the models only require that highly likely errors need to be taken into account:
in the present paper, 
we assume the transmission channel modeled by some {\it binary word relation}, namely $\tau\subseteq A^*\times A^*$.
In order to retrieve $u$, the homomorphism $\phi$, and thus the code $X$, must satisfy specific
constraints, which of course depend on the channel $\tau$:
in view of some  formalization, we denote by $\widehat\tau$ the {\it reflexive closure} of $\tau$, and by  $\uline\tau$ its  {\it anti-reflexive restriction} that is, 
$\tau\setminus\{(w,w):w\in A^*\}$.

About the channel itself, 
the so-called 
{\it synchronization constraint} appears mandatory: it sets that,
for each input word  factorized $w=x_1\cdots x_n$, where $x_1,\cdots,x_n$ are {\it codewords} in $X$,
every output word has to be factorized 
$w'= x'_1\cdots x'_n$, with $(x_1,x'_1),\cdots, (x_n,x'_n)\in\widehat\tau$. 
In order to ensure such a constraint,  as for the Morse code, some pause symbol could be inserted after each codeword $x_i$.

Regarding the code $X$,  in order to minimize the number of errors, in most cases 
some close neighborhood constraint is applied over  $\widehat\tau(X)$, the set of the images of the codewords under $\widehat\tau$. 
In the most frequent use, such a constraint consists of some minimal distance condition:
the smaller the distance between the input codeword $x\in X$ and any of its  corresponding output words $x'\in \widehat\tau(X)$,
the more optimal is error detection. 
In view of that, we fix over $A^*$ a {\it quasi-metric} $d$, in the meaning of \cite{W31} (the difference with a {\it metric} is that $d$ needs not to satisfy the symmetry axiom).

As outlined in \cite{CP02}, given an error tolerance level  $k\ge 0$, a corresponding  binary word relation, denoted in the present paper by  $\tau_{d,k}$,  can be associated in such a way that 
we set $(w,w')\in\tau_{d,k}$ (or equivalently, $w'\in \tau_{d,k}(w)$), whenever $d(w,w')\le k$ holds.

Below, in the spirit of \cite{JK97,N21}, we draw some specification regarding  error detection capability.
Recall that a subset  $X$ of $A^*$ is {\it independent} with respect to $\tau\subseteq A^*\times A^*$ (for short, $\tau$-{\it independent})
iff.  $\tau(X)\cap X=\emptyset$ holds. This notion, which appears dual with the one  of  {\it closed set} \cite{N21},  relies  to the famous {\it dependence systems} \cite{C81,JK97}.
In addition, given a family of  codes, say ${\cal F}$, a code  $X\in\cal F$  is {\it maximal in $\cal F$}
whenever $X\subseteq Y$, with $Y\in {\cal F}$, implies $Y=X$.

Given a code $X\subseteq A^*$, we introduce the four following conditions:
\begin{enumerate} [label={\bf                                                                                                                                                                                                                                                                                                                                                                                                                                                                      (c\arabic*)},wide=10pt]
\item \label{1} 
{\it Error detection:}
 $X$ is $\uline{\tau_{d,k}}$-{\it independent}. 
\item
\itemsep2pt
\label{2}
{\it Error correction:}
$x,y\in X$ and $\tau_{d,k}\left(x\right)\cap\tau_{d,k}\left(y\right)\neq\emptyset$ implies  $x=y$.
\item \label{3}   {\it $X$ is maximal in the family of \uline{$\tau_{d,k}$}-independent  codes}.
\item \label{4}  {\it $\widehat{\tau_{d,k}}(X)$ is a code.}
\end{enumerate}
A few  comments on Conds. \ref{1}--\ref{4}:

 \noindent\hspace*{1mm}
-- By definition, Cond. \ref{1} is satisfied 
 iff.  the quasi-distance between pairs of different elements of $X$ is greater than $k$
that is,  $X$ is able to detect at most $k$ errors in the transmission of any codeword.

\noindent\hspace*{1mm}
-- Cond. \ref{2}  sets a classical definition: it expresses that if some transmission error has been detected in an output word, necessarily such a word comes from a unique input codeword.
 
\noindent\hspace*{1mm}
-- With Cond. \ref{3}, in the family of \uline{$\tau_{d,k}$}-independent codes, $X$ cannot be improved.
From this point of view,  fruitful investigations have been done in several famous classes determined by code properties \cite{JKK01,KKK14,L00,L01}.
%

\noindent\hspace*{1mm}
-- At last, Cond. \ref{4} 
expresses that the factorization of any output
message over the set  $\widehat{\tau_{d,k}}(X)$ is done in a unique way. 
Actually, since $d$ is a quasi-metric, the corresponding relation ${\tau_{d,k}}$ is reflexive, therefore Cond. \ref{4} is equivalent to ${\tau_{d,k}}(X)$ itself being a code.
\smallskip

Actually, in most of the cases it could be very difficult, even impossible, to satisfy all together Conds. \ref{1}--\ref{4}: 
for instance, as shown in \cite{JK97,N21}, there are regular codes satisfying \ref{1} that cannot satisfy  \ref{2}.
Furthermore some compromise has to be adopted:
in view of this, given a regular code $X$, a natural question consists in examining whether each of  those conds. is satisfied  in the frameworks of classical free monoid quasi-metrics.
From this point of view, in \cite{N21}, we considered  the so-called  {\it edit relations},
 some peculiar  compositions of one-character {\it deletion}, {\it insertion}, and {\it substitution}:
such relations involve  the famous Levenshtein and Hamming metrics \cite{H50,K83,L65}, which  are prioritary related to {\it subsequences} in words.
In the present paper, we  focuse on the following quasi-metrics, the two first  ones involving {\it factors}:

\smallskip
\noindent\hspace*{1mm}
-- The {\it prefix} metric is defined  by $d_{\rm P}(w,w')=|w|+|w'|-2|w\wedge w'|$,
where $|w|$ stands for the length of the word $w$, and  $w\wedge w'$ denotes  the maximum length common {\it prefix} of $w$ and $w'$:
we set ${\cal P}_k=\tau_{d_{\rm P},k}$.

\noindent\hspace*{1mm}
-- The {\it factor} metric, for its part,  is
defined by $d_{\rm F}(w,w')=|w|+|w'|-2|f|$, where $f$ is a maximum length  common factor of $w$, $w'$: we set ${\cal F}_k=\tau_{d_{\rm F},k}$.

\noindent\hspace*{1mm}
-- A third type of topology can be introduced in connection with {\it monoid automorphisms} or {\it anti-automorphisms} (for short, we write {\it (anti-)automorphisms}):
such a topology particularly involves the domain of  DNA sequence comparison. 
By anti-automorphisms of the free monoid, we mean any one-to-one mapping onto $A^*$, say $\theta$, 
st.   the equation  $\theta(uv)=\theta(v)\theta(u)$ holds for  any $u,v\in A^*$: as in the case of automorphisms, each of those mappings actually extends to $A^*$ some permutation of $A$
(for involvements in the framework of closed codes, see \cite{NS20}).
With every  (anti-)automorphism $\theta$ we associate  the quasi-metric $d_\theta$, defined as follows:

\smallskip
\noindent\hspace*{2mm}(1) $d_\theta(w,w')=0$ is equivalent to $w=w'$;

\noindent\hspace*{2mm}(2) we set $d_\theta(w,w')=1$ whenever  $w'=\theta(w)$ holds,  with $w\ne w'$;

\noindent\hspace*{2mm}(3)  in all other cases we set $d_\theta(w,w')=2$. 

\smallskip\noindent 
It can be easily verified that, $k\ge 2$ implies $\tau_{d,k}=A^*\times A^*$. In other words, wrt. error detection constraints only the cond. $k=1$ takes sense:
by definition we have $\tau_{{d_\theta},1}=\widehat{\theta}$ and $\uline{\tau_{_\theta,1}}=\uline\theta$.

\smallskip
This paper  relates an extended and augmented version of the study  we presented in \cite{N22}.
In particular, answers are provided to some of the open questions that were asked: they concern the behavior of regular codes  wrt.   ${\cal F}_k$.
We prove the following result:
{\flushleft
{\bf Theorem.}}
{\it With the preceding notation, given a regular code $X\subseteq A^*$, for every $k\ge 1$, 
it can be decided whether $X$ satisfies any of Conds.  \ref{1}--\ref{4} wrt.  ${\cal P}_k$, ${\cal F}_k$, and $\widehat\theta$.
}

\medskip\noindent  Some comments about the proof:

\noindent
\hspace*{1mm}-- Regarding Cond. \ref{1}, 
we establish that, for each of the mentioned quasi-metrics, $\uline{\tau_{d,k}}(X)$
 is a  {\it regular} subset of $A^*$. When $d$ is  the prefix metric, this is done by proving that $\uline{\tau_{d,k}}$ itself is a regular relation that is, a regular subset of the monoid $A^*\times A^*$
 in the sense of \cite{EM65}.

 In the case where $d$ corresponds to the factor metric, although the question of the regularity of $\uline{\tau_{d,k}}$ still remains open,
 the result is obtained thanks to the construction of a peculiar finite set covering for $\uline{{\cal F}_k}\subseteq A^*\times A^*$.
Actually $\uline{{\cal F}_k}$ is precisely the union of the sets in that finite family.
Moreover the so-called {\it conjugacy}, some concepts from combinatorics on words \cite{Lo1983} allows to prove that each of those sets is  regular.
Regarding $\widehat\theta=\tau_{d_\theta,1}$, we prove   that   in any case $X$ satisfies Conds. \ref{1}, \ref{2} .

\noindent\hspace*{1mm}
--  In the case of the  relation ${\cal P}_k$ (resp., ${\cal F}_k$), we prove that $X$ satisfies Cond. \ref{2} iff.  it satisfies Cond. \ref{1} wrt.  
${\cal P}_{2k}$ (resp., ${\cal F}_{2k}$). 

\noindent\hspace*{1mm}
-- Wrt. each of the quasi-metrics raised in the paper we established  that, given a regular code $X\subseteq A^*$, $X$ is maximal in the family of the codes independent wrt. $\uline{\tau_{d,k}}$
iff. it is {\it complete} that is, every word of $A^*$ is  a factor of some word in $X^*$, the free submonoid of $A^*$ generated by $X$. 
Actually this is done by proving that    
any non-complete  $\uline{\tau_{d,k}}$-independent code  can be embedded into some complete one: 
in other words it cannot be maximal.
 In order to establish such a property, 
in the spirit of  \cite{BWZ90,L03,N06,N08,NS20,ZS95}, we provide specific regularity-preserving embedding formulas:

their schemes are based upon the methodology from \cite{ER85}.
Notice that,  in   \cite{JKK01,KM15,L00,L01,VVH05}, wrt.   peculiar families of sets, algorithmic methods for embedding a set into  some maximal (but not necessarily complete) one were also provided.

\noindent\hspace*{1mm}
-- Regarding Cond. \ref{4}, for each of the preceding relations, the set $\widehat{\tau_{d,k}}(X)={\tau_{d,k}}(X)$ is regular, therefore in any case, 
by applying the famous Sardinas and Patterson algorithm \cite{SP53}, one can decide whether that cond. is satisfied.

\medbreak
\noindent We now shorty describe the contents of the paper:

\noindent\hspace*{5mm}
-- Section 2 is devoted to the preliminaries: we recall fundamental notions about words, word binary relations,  regular
sets, and codes. 
Taking account that $d$ is a quasi-metric,  wrt. $\tau_{d,k}$ some equivalent formulations of the error correction cond. \ref{2} are provided.

\noindent\hspace*{5mm}
-- The aim of Sect. 3 is to study the relation ${\cal P}_k$. 
In addition,  corresponding  results wrt.  the so-called  {\it suffix} metric are set.

\noindent\hspace*{5mm}
-- Sect. 4 is devoted to a preliminary study about ${\cal F}_k$  and $\uline{{\cal F}_k}$.

We construct the above-mentioned finite covering  for $\uline{{\cal F}_k}\subseteq A^*\times A^*$.

Futhermore  we prove that  $\uline{{\cal F}_k}$ is regularity-preserving.

\noindent\hspace*{5mm}
-- The decidability results involving the factor metric are established in Sect. 5.

\noindent\hspace*{5mm}
-- Sect. 6 is devoted to quasi-metrics associated to (anti-)automorphisms. 

\noindent\hspace*{5mm}
-- The paper concludes  with some possible directions for further research (Sect. 7). 

\noindent\hspace*{5mm}
-- In order to make the 
decidability results clearer, if needed,
an appendix is added at the end of the paper. It  provides some basic support in order to further implementing corresponding algorithms.
Regarding the main study, in no way that appendix can constitute any prerequisite.
\section{Preliminaries}
\label{prelim}
Several definitions and notations has already been settled.
In what follows, we bring precision about concepts such as words,  automata,  regular relations and variable-length codes. If necessary, we suggest the reader that he (she)  report to classical books such as \cite{BPR10,E74,HU79,S03}.
Some classical decidability results are also set (for corresponding schemes of implementation see the appendix).

\subsection{Words}
In the whole paper, we fix a finite alphabet $A$, with $|A|\ge 2$. We denote by $\varepsilon$ 
the {\it empty word} that is, the word with {\it length} $0$, and we set $A^+=A^*\setminus\{\varepsilon\}$.
Given two words $v,w\in A^*$, $v$ is a {\it prefix} (resp., {\it suffix}, {\it factor}) of $w$ if 
words $u,u'$ exist st.  $w=vu$ (resp., $w=u'v$, $w=u'vu$).
 In the case where the equation $w=uv$ holds, we set $u=wv^{-1}$ and  $v=u^{-1}w$.
We denote by ${\rm P}(w)$ (resp.,  ${\rm S}(w)$, ${\rm F}(w)$) the set of the words that are prefixes (resp., suffixes,  factors) of $w$.
In the case where we have $v\ne w$, with $v\in {\rm P}(w)$ (resp., $v\in{\rm S}(w)$), we say that $v$ is a {\it proper prefix} (resp., {\it proper suffix}) of $w$.
More generally,  given $X\subseteq A^*$, we denote by ${\rm P}(X)$ the union of the sets ${\rm P}(x)$, for all the words $x\in X$ 
(the sets  ${\rm S}(X)$ and ${\rm F}(X)$ are defined in a similar way).
Given a word $w\in A^*$, we denote by $w^R$ its {\it reversal} that is, for $a_1,\cdots,a_n\in A$,  we have $w^R=a_n\cdots a_1$ whenever $w=a_1\cdots a_n$ holds.

%

A word $w\in A^*$  is {\it overlapping} whenever some $v\in A^*$ exists st.  $wv\in A^*w$, with $1\le |v|\le |w|-1$;
otherwise, $w$ is {\it overlapping-free} that is,  $wv\in A^*w$ with $|v|\le |w|-1$ implies $v=\varepsilon$.
For instance, $w=ababa$ is overlapping: taking $v=ba$, we have $wv=(ababa)(ba)=(ab)(ababa)\in A^*w$;
on the contrary, $w=ababb$ is overlapping-free.
Note that $w$ is overlapping-free iff. $w^R$  itself is overlapping free.
Classically, the following property holds (see e.g.  \cite[Proposition 1.3.6]{BPR10}):
\begin{proposit}
\label{overlapping-free constr}
Given $z_0\in A^+$, let $a$ be the initial letter of $z_0$ and  let $b\in A\setminus\{a\}$.
Then the word $z_0ab^{|z_0|}$ is overlapping free.
\end{proposit}
%
\noindent
Two words $w,w'\in A^+$ are {\it conjugate} iff. a pair of words $\alpha$, $\beta$ exist st.  $w=\alpha\beta$ and $w'=\beta\alpha$, with $\beta\ne\varepsilon$.
The following result brings additional information: 
\begin{proposit}
\label{conjugacy}
{\rm \cite[Proposition 1.3.4]{Lo1983}}
Given a pair of non-empty words  $w,w'$, the two following  conds. are equivalent:

{\rm (i)} $w$ and $w'$ are conjugate.

{\rm (ii)} $t\in A^*$ exists st.  $wt=tw'$.
More precisely $\alpha\in A^*$, $\beta\in A^+$, and $n\in {\mathbb N}$ exist st.  $w=\alpha\beta$, $w'=\beta\alpha$, and $t\in(\alpha\beta)^n\alpha$.
\end{proposit}
\subsection{Words binary relations}
\noindent 
Let $M$ be an arbitrary monoid. The following basic concepts are involved by our study:
%
\paragraph {Operations among subsets of $M$}
 Given two sets $X,Y\subseteq M$, their {\it concatenation product} is  $XY=\{xy: x\in X, y\in Y\}$
and the {\it Kleene star} of $X$ is $X^*=\{x_1 \cdots x_n:x_1,\cdots x_n\in X, n\ge 0\}$. 
Union, concatenation product, and Kleene star constitute the so-called {\it regular operations} into $2^M$.
In addition, the {\it left quotient} (resp., {\it right quotient}) of $X$ by $Y$ is $Y^{-1}X=\{z\in M : (\exists x\in X) ,  (\exists y\in Y),  x=yz\}$ (resp.,  $XY^{-1}=\{z\in M : (\exists x\in X) ,  (\exists y\in Y),  x=zy\}$).
\paragraph{$M$-automata} 
A $M$-automaton, say ${\cal A}$, is defined on the basis of a  finite labelled graph. 
The vertices are the so-called {\it states} and the $M$-labeled edges are the {\it transitions}: let $Q$ and $E\subseteq Q\times M\times Q$ be the corresponding sets.
The automaton is {\it finite} whenever  $E$ is a finite set.
Two subsets of $Q$ are actually distinguished  namely $I$, the initial states, and $T$, the terminal ones. 
In the paper, the transition $(q,m,q')\in E$
is commonly  denoted by  $q\xrightarrow{m}q'$.
Denoting by $1_M$ the identity element in the monoid $M$,  for any $q\in Q$, $q\xrightarrow{1_M}q$ is a transition:
 accordingly we left it out in  any graphical representation of ${\cal A}$.
A {\it successful path} is a chain
$(i\xrightarrow{m_0}q_1,q_1\xrightarrow{m_1}q_2,  \cdots, q_n\xrightarrow{m_n} t)$, with $i\in I$ and $t\in T$. For convenience, we denote it by $ i\xrightarrow{m_0}q_1\xrightarrow{m_1}q_2\cdots q_n\xrightarrow{m_n} t$.
The  so-called {\it behavior} of ${\cal A}$, which we denote by $\left|{\cal A}\right|$, is the subset with elements all the labels $m_0m_1\cdots m_n\in M$  of the corresponding successfull  paths.
Classically, for any finite $A^*$-automaton ${\cal A}$  there is another finite automaton ${\cal A'}$, with transitions labeled by characters of $A$, st. $\left|{\cal A}\right|=\left|{\cal A'}\right|$.
\paragraph{Regular sets}  A set $X\subseteq M$ is {\it regular} (or equivalently {\it rational}) if  it belongs to the {\it regular closure} of the finite subsets of $M$ that is, the smallest (wrt.  the inclusion) subset of $2^M$ that contains the finite subsets and which is  closed under the  regular operations (see \cite[Sect. II.1]{S03}).
By definition, the family of regular sets is closed under the regular operations.
The following property is attributed to Elgot and Mezei \cite{EM65}:
\begin{theorem} 
\label{Elgot-Mezei-0}
Let $M$ be  a monoid and  $X\subseteq M$.  The following conds. are equivalent:

{\rm (i)} $X$ is regular.

{\rm (ii)} $X$ is the behavior of some finite $M$-automaton.
\end{theorem}
\begin{example}
{\rm Let $M=\{a\}^*\times \{b\}^*$, and $X=\{(a^{n+1},b^{2n}): n\ge 0\}$. We have $X=YZ$, with $Y=\{(a,b^2)\}^*$ and $Z=\{(a,\varepsilon)\}$, whence $X$ is regular.
Actually, $X$ is the behavior of the finite $M$-automaton represented in Figure \ref{example-reg-subset-M}.
\begin{figure}[H]
\begin{center}
\includegraphics[width=6.28cm,height=1.64cm]{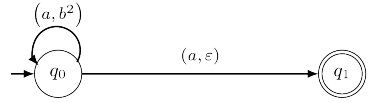}
\end{center}
\caption[]{\small A $M$-automaton with behavior $X$: 
by convention, the initial state is identified with an input arrow and the terminal one  with a double circle; transitions with label $1_M$ are left out.}
\label{example-reg-subset-M}
\end{figure}
}
\end{example}
\noindent In what follows, we recall some classical closure properties of regular sets:
\begin{proposit} 
\label{property-reg}
Given a pair of  monoids $M$, $N$,   and a   monoid homomorphism $h:M\rightarrow N$, the image of every regular subset of $M$ under $h$ is a regular subset of $N$.
\end{proposit}
\noindent  In the most general case the family of regular subsets of $M$ is not closed neither under     intersection, nor complementation, nor inverse monoid homomorphism.
 However, in the case where $M=A^*$, the following noticeable result holds:
\begin{proposit}
\label{properties-reg-A}
The family of regular subsets of $A^*$ is closed under boolean operations, regular operations, left-quotient (resp. right-quotient),  and direct (resp. inverse) monoid endomorphism.
\end{proposit}
\noindent In particular the finite union (resp., intersection, concatenation) of regular sets is itself regular.

\paragraph{Regular relations}  Let  $M,N$ be two monoids.
A binary relation from $M$ into $N$ consists in any subset  $\tau$ of $M\times N$.
Since $M\times N$ itself is a monoid, 
those relations are directly involved by the preceding concept of regularity.
The {\it composition} in this order of  $\tau$ by $\tau'$ is defined by  $\tau\cdot\tau'\left(x\right)=
\tau'\left(\tau\left(x\right)\right)$ (the notation  $\tau^k$ refers to that operation).
The  {\it inverse} of $\tau$ is  the relation $\tau^{-1}\subseteq N\times M$  defined by  $(w,w')\in\tau^{-1}$ whenever $(w',w)\in\tau$; in addition
its {\it complement}  is  $\overline\tau=M\times N\setminus \tau$.
Composition and inverse preserve regularity among relations; however
in the most general case a regular relation is not preserved under complementation.  
In the paper the following result (see \cite[Sect. IX.3]{E74}  or  \cite[Sect. IV.1.3]{S03}) will be frequently applied:
\begin{proposit} 
\label{property-rec1} 
 Given a regular relation $\tau\subseteq A^*\times A^*$,
and a regular  set $X\subseteq A^*$, the set $\tau(X)$ is  regular.
\end{proposit}
\noindent At last, $id_{A^*}=\{(w,w)|w\in A^*\}$ and its complement  $\overline {id_{A^*}}$ are regular relations 
(see Fig. \ref{Automaton-comp-idA*} ).
\begin{figure}[H]
\begin{center}
\includegraphics[width=4.5cm,height=5.5cm]{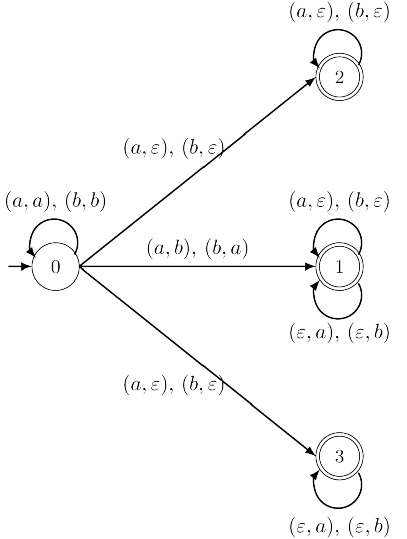}
\end{center}
\caption[]{\small An automaton with behavior $\overline{id_{A^*}}$ in the case where $A=\{a,b\}$ (arrows are multi-labelled).}
\label{Automaton-comp-idA*}
\end{figure}
\subsection{Variable-length codes}
\noindent Given a subset $X$ of $A^*$, and $w \in X^*$, let $x_1, \cdots, x_n\in X$ st.  $w$ is the result of the concatenation of the words $x_1, x_2,\cdots, x_n$, in this order. 
In view of specifying the factorization of $w$ over X, we use the notation $w =(x_1)(x_2)\cdots (x_n)$, or equivalently: $w =x_1\cdot x_2\cdots x_n$. 
For instance, over the set $X=\{a, ab, ba\}$, the word $bab\in X^*$ can be factorized as  $(ba)(b)$ or $(b)(ab)$ (equivalently denoted by $ba\cdot b$ or  $b\cdot ab$).

A set $X$ is a {\it variable-length code} (a {\it code} for short) if  for any pair of finite sequences of words in $X$, say  $(x_i)_{1\le i\le n}$, $(y_j)_{1\le j\le p}$, the equation
$x_1\cdots x_n=y_1\cdots y_p$ implies  $n=p$, and $x_i=y_i$ for each integer $i\in [1,n]$.
In other words, every element of $X^*$ has a unique factorization over $X$, equivalently the submonoid $X^*$ is {\it free}. In particular, 
$X\ne\{\varepsilon\}$ is a {\it prefix} (resp., {\it suffix}) {\it code} whenever, for every pair of words $x,y\in X$, the cond. $x\in {\rm P}(y)$ (resp., $x\in {\rm S}(y)$) implies $x=y$.
In addition, $X$ is a {\it bifix} code  if it is both a prefix code and a suffix one. Any set $X\ne\{\varepsilon\}$ whose elements have a common length is a complete  bifix code,
we say that it is a {\it uniform code}.

A positive {\it Bernoulli distribution} consists in a total mapping $\mu$ from the alphabet $A$ into ${\mathbb R}_+=\{x\in {\mathbb R}: x\ge 0\}$ (the set of the non-negative real numbers)  st.    $\sum_{a\in A}\mu(a)=1$.
Such a mapping  is extended into a unique  monoid homomorphism from  $A^*$ into $({\mathbb R}_+,\times)$, which is itself extended into a unique  positive {\it Bernoulli measure}  $\mu:2^{A^*}\longrightarrow {\mathbb R}_+$.
 In order to do so, for each word $w\in A^*$, we set  $\mu\left(\{w\}\right)=\mu(w)$; in addition given two disjoint subsets $X,Y$ of $A^*$, we set $\mu(X\cup Y)=\mu(X)+\mu(Y)$.
According to the famous {\it  Kraft inequality},  for every
 code $X$ we have $\mu(X)\le 1$.
In the whole paper, we take for $\mu$  the  so-called {\it uniform} Bernoulli measure: it is determined by $\mu(a)=1/|A|$, for each $a\in A$. 

The two following results are classical: the first one is due to Sch\"utzenberger (see \cite[Theorem 2.5.16]{BPR10}) and the second provides some answer to a question actually set in \cite{R77}.
\begin{theorem}
\label{classic}
Given a  regular code $X\subseteq A^*$,   the following conds. are equivalent:

{\rm (i)} $X$ is {\it complete}.

{\rm (ii)} $X$ is a maximal  code.

{\rm (iii)}    $\mu(X)=1$ holds.
\end{theorem}
\begin{theorem}{\rm \cite{ER85}}.
\label{EhRz}
Let $X\subseteq A^*$ be a non-complete
code and let $z\notin {\rm F}(X^*)$ be an overlapping-free word. Set $U=A^*\setminus\left( X^*\cup A^*zA^*\right)$ and $Y=(zU)^*z$. Then $Z=X\cup Y$ is a   
complete code.
\end{theorem}
\noindent With the cond. of Theorem \ref{EhRz}  if $X$ is   regular then the same holds for the resulting code $Z$.
In addition, since $z$ is an overlapping-free word in  $A^*\setminus{\rm F}(U)$, the following property holds: it will be applied further in the paper
(see  Prop.  \ref{Comp-Pref-Independent}).
\begin{lem}
\label{prefix-Uz}
With the preceding notation, the set $Uz$ is a prefix code.
\end{lem}

\noindent Finally, the following result is the basis of the decidability properties we establish in the paper: %
\begin{proposit}
\label{decidable-bases}
Given a monoid  $M$ and a regular set $X\subseteq M$, what follows holds:

{\rm (i)} In any case it can be decided whether $X=\emptyset$.

{\rm (ii)} If $M=A^*$ then it can be decided whether $X$ is a code.

{\rm (iii)} If $M=A^*$ then one can decide whether $\mu(X)=1$ holds.
\end{proposit}
\subsection{Some equivalent formulation of the error correction condition.}
\label {deeper-quasi-metrics} 

Let  $A$ be an alphabet and $d: A^*\times A^*\longrightarrow {\mathbb R}^+$ be a quasi-metric over $A^*$.  
Given a positive integer $k$, by definition we have $\tau_{d,k}\subseteq \tau_{d,k+1}$. In addition, if $d$ is a metric then $\tau_{d,k}$ is a symmetric relation that is, the equation $\tau_{d,k}^{-1}=\tau_{d,k}$ holds.
We close Sect. \ref{prelim} by proving the following result, which will be applied several times in the rest of the paper:
\begin{lem}
\label{equiv-c2}
Given  a quasi-metric $d$ over $A^*$ and $X\subseteq A^*$, the following conds. are equivalent:

{\rm (i)} $X$ satisfies Cond. \ref{2} wrt.  $\tau_{d,k}$.

{\rm (ii)}  For every $x\in X$ we have $\tau_{d,k}^{-1}\left(\tau_{d,k}\left(x\right)\right)\cap X=\{x\}$.

{\rm (iii)} $X$ satisfies Cond. \ref{1} wrt.  $\tau_{d,k}\cdot\tau_{d,k}^{-1}$.
\end{lem}
\begin{proof}
-- Beforehand, we
prove that any  singleton satisfies each of Conds. (i)--(iii) of the lemma.
Let  $X=\{x\}$.  By definition, trivially $X$  satisfies Cond. \ref{2}.
In addition, since $\tau_{d,k}$ is reflexive,  we have  
$x\in \tau_{d,k}^{-1}\left(x\right)$, thus $x\in \tau_{d,k}\cdot\tau_{d,k}^{-1}\left(x\right)$. We obtain  $\tau_{d,k}\cdot\tau_{d,k}^{-1}(X)\cap X=\{x\}$, whence $X$ satisfies Cond. (ii). 
Finally, $\tau_{d,k}\cdot\tau_{d,k}^{-1}\left(x\right)\cap X=\{x\}$ implies 
$\uline{\tau_{d,k}\cdot\tau_{d,k}^{-1}}(X)\cap X=\emptyset$, thus $X$ satisfies  Cond.\ref{1} wrt. $\uline{\tau_{d,k}\cdot\tau_{d,k}^{-1}}$.
In the rest of the proof we assume $|X|\ge 2$.

\smallbreak
\noindent \hspace*{2mm}-- In what follows we prove that Conds. (i) and (ii) are equivalent.
Firstly, assuming that $X$ satisfies Cond. (i),  
consider $x\in X$ and  
 $y\in\tau_{d,k}^{-1}\left(\tau_{d,k}\left(x\right)\right)\cap X$.
By construction $\tau_{d,k}\left(x\right)\cap\tau_{d,k}\left(y\right)\neq\emptyset$ holds, therefore we have  $x=y$, thus $\tau_{d,k}^{-1}(\tau_{d,k}\left(x\right))\cap X=\{x\}$.
Conversely, assuming that Cond. (ii) holds, 
let $x,y\in X$ st.  $\tau_{d,k}\left(x\right)\cap\tau_{d,k}\left(y\right)\ne\emptyset$. Some $z\in \tau_{d,k}\left(x\right)\cap\tau_{d,k}\left(y\right)$ exists, moreover
we have $y\in\tau_{d,k}^{-1}(z)\subseteq\tau_{d,k}^{-1}(\tau_{d,k}\left(x\right))$. It follows from $y\in X$ and 
 $\tau_{d,k}^{-1}\left(\tau_{d,k}\left(x\right)\right)\cap X=\{x\}$ that $y=x$.

\smallbreak
\noindent \hspace*{2mm}--  We prove that Cond. (i) implies Cond.  (iii) in  arguing by contrapositive.
Assuming that $X$ does not satisfy Cond. (iii), by definition  we have $\uline{\tau_{d,k}\cdot\tau_{d,k}^{-1}}(X)\cap X\ne\emptyset$, whence
$x,y\in X$  exist st.  $y\in\uline{\tau_{d,k}\cdot\tau_{d,k}^{-1}}\left(x\right)$. 
By definition, we have $y\in\tau_{d,k}\cdot\tau_{d,k}^{-1}\left(x\right)$, with $y\ne x$.
In other words $z\in A^*$ exists  st.  we have $z\in\tau_{d,k}\left(x\right)$ and $y\in\tau_{d,k}^{-1}(z)$ that is,
$z\in \tau_{d,k}\left(x\right)\cap\tau_{d,k}\left(y\right)$, thus $ \tau_{d,k}\left(x\right)\cap\tau_{d,k}\left(y\right)\ne\emptyset$.
It follows from $y\ne x$ that $X$ cannot satisfy Cond. \ref{2} wrt. $\tau_{d,k}$.

\smallbreak
\noindent \hspace*{2mm}--  Once more arguing by contrapositive, we prove that Cond. (iii) implies Cond. (ii). 
Assume that $x\in X$ exists st. $\tau_{d,k}^{-1}\left(\tau_{d,k}\left(x\right)\right)\cap X\ne\{x\}$. There is some $y\ne x$ st. 
$y\in \uline{\tau_{d,k}\cdot\tau_{d,k}^{-1}}(X)\cap X$: by definition $X$ cannot be  $\uline{\tau_{d,k}\cdot\tau_{d,k}^{-1}}$-independent.
\end{proof}
\section{Error detection and the prefix metric}
\label{Prefix-metric}

We start with a few examples:

\begin{example}
\label{Prefs2}
{\rm Over $A=\{a,b\}$, consider the finite prefix  code  $X=\{a,ba,b^2\}$.

\noindent\hspace*{2mm} -- $ X$ satisfies Cond. \ref{1} wrt.   ${\rm \cal P}_1$ (that is,  $X$ is $1$-error-detecting).
Indeed, it follows from $a\wedge ba=a\wedge b^2=\varepsilon$ and $ba\wedge b^2=b$ that $d_{\rm P}(a,ba)=|a|+|ba|=3$, $d_{\rm P}(a,b^2)=|a|+|b^2|=3$, and $d_{\rm P}(ba,b^2)=|ba|+|b^2|-2|b|=2$.
Consequently,  for each pair of different words $x,y\in X$, we have $(x,y)\notin {\cal P}_1$ that is,  $\uline{{\rm \cal P}_1}(X)\cap X=\emptyset$.

\noindent  \hspace*{2mm}--  Regarding Cond. \ref{2}, in view of Lemma \ref{equiv-c2}, firstly we compute 
$\uline{{\cal P}_1\cdot{\cal P}_1^{-1}}\left(ba\right)\cap X=\uline{{\cal P}_1^{2}}\left(ba\right)\cap X$.
It follows from ${\cal P}_1\left(ba\right)=\{b,ba,ba^2,bab\}$
that  ${\cal P}_1^{2}\left(ba\right)=\{\varepsilon,b,ba,b^2,ba^2,bab,ba^3,\\ ba^2b,baba,bab^2\}$, 
thus $\uline{{\cal P}_1^{2}}\left(ba\right)=\{\varepsilon,b,b^2,ba^2,bab,ba^3,ba^2b, baba,bab^2\}$. This implies
$\uline{{\cal P}_1^{2}}\left(ba\right)\cap X=\{b^2\}$, whence $X$ cannot satisfy Cond. \ref{1} wrt.  $\uline{{\cal P}_1\cdot{\cal P}_1^{-1}}$. According to Lemma \ref{equiv-c2},
$X$ cannot satisfy Cond. \ref{2} wrt. ${\cal P}_1$.

\noindent  \hspace*{2mm} --  We have  $\mu(X)=1/2+1/4+1/4=1$ therefore, according to  Theorem \ref{classic}, $X$ is a maximal code. Since it is $\uline{{\cal P} _1}$-independent,  $X$ is maximal   in the family of $\uline{{\cal P}_1}$-independent codes that is, it satisfies Cond. \ref{3}.

\noindent  \hspace*{2mm} -- Since we have $X\subsetneq {\cal P}_1(X)$, and since $X$  is a maximal code, the set $\widehat {{\cal P}_1}(X)={\cal P}_1(X)$ cannot be a code that is, $X$ cannot satisfy Cond. \ref{4} (we verify that we have $\varepsilon\in{\cal P}_1(X)$).
}
\end{example}
\begin{example}
\label{Prefs3}
{\rm 
Let  $n\ge 2$ and $k\in [1,n-1]$. Consider the   uniform code $X=A^n$. 

\noindent  \hspace*{2mm} 
-- For any $x\in X$ we have   $\uline{{\cal P}_k}(x)\subseteq \left(A^{n-k}\cup\cdots\cup A^{n-1}\right)\cup \left(A^{n+1}\cup\cdots\cup\\ A^{n+k}\right)$:
this implies $\uline{{\cal P}_k}(X)\cap X=\emptyset$, thus $X$ satisfies Cond. \ref{1} wrt.  ${\cal P}_k$. 

\noindent  \hspace*{2mm} --  However, Cond. \ref{2} is not satisfied by $X$:
indeed, given two different characters  $a,b$, we have $a^{n-1}\in{\cal P}_k(a^n)\cap {\cal P}_k(a^{n-1}b)$.

\noindent  \hspace*{2mm} --  As mentioned in the preliminaries $X$ is complete. According to Theorem \ref{classic} $X$ is a maximal (bifix) code,
hence it is maximal in the family of $\uline{{\cal P}_k}$-independent codes that is, $X$ satisfies  Cond. \ref{3} wrt. ${\cal P}_k$.

\noindent  \hspace*{2mm}--  We have $X\subsetneq {\cal P}_k(X)$: since $X$ is a maximal code, it  cannot satisfies Cond. \ref{4}.

}
\end{example}
\begin{example}
\label{Prefs4}
{\rm 
Over the alphabet $A=\{a,b\}$, consider the  regular bifix code $X=\{ab^na: n\ge 0\}\cup \{ba^nb: n\ge 0\}$ and the relation ${\cal P}_1$.

\noindent  \hspace*{2mm} --   It follows from ${\uline{\cal P}_1}(X)=\bigcup_{n\ge 0}\{ab^n,ab^na^2,ab^nab, ba^n, ba^nba,  ba^nb^2\}$,
 that ${\uline{\cal P}_1}(X)\cap X=\emptyset$, whence $X$ satisfies Cond. \ref{1}.

\noindent  \hspace*{2mm} --   For $n\ne 0$ we have 
 $\uline{{\cal P}_1\cdot{\cal P}_1^{-1}}\left(ab^na\right)={\cal P}_1^{2}\left(ab^na\right)\setminus\{ab^na\}=\{ab^n, ab^{n-1}, ab^{n+1}, ab^na^2, ab^na^3,\\ab^na^2b,ab^nab,ab^naba,ab^nab^2\}$. 
Similarly, for $n=0$ we have $\uline{{\cal P}_1\cdot{\cal P}_1^{-1}}\left(ab^na\right)= \uline{{\cal P}_1}^2(a^2)={\cal P}_1\left(\{a,a^2,a^3,a^2b\}\right)\setminus\{a^2\}=\{\varepsilon, a, ab, a^3,  a^2b, a^4, a^3b, a^2b,a^2ba,a^2b^2\}$.
In any case we obtain $X\cap \uline{{\cal P}_1\cdot{\cal P}_1^{-1}}\left(ab^na\right)=\emptyset$.
 Similarly,  we have $X\cap \uline{{\cal P}_1\cdot{\cal P}_1^{-1}}\left(ba^nb\right)=\emptyset$.
Consequently, $X$ is $\uline{{\cal P}_1\cdot{\cal P}_1^{-1}}$-independent therefore,
according to Lemma \ref{equiv-c2}, $X$ satisfies Cond. \ref{2} wrt.  ${\cal P}_1$.

\noindent  \hspace*{2mm} --  Regarding Cond. \ref{3}, we have $\mu(X)=2\cdot1/4\sum_{n\ge 0}(1/2)^n=1$:
according to Theorem \ref{classic}, $X$ is a maximal code, whence it is maximal in the family of $\uline{{\cal P}_1}$-independent codes.

\noindent  \hspace*{2mm} --  Since we have $X\subsetneq {\cal P}_1(X)$,  $X$ cannot satisfies Cond. \ref{4} (we verify that $a,a^2\in {\cal P}_1(X)$).
}
\end{example}
\subsection{A preliminary study of the relation ${\cal P}_k$}
Given a pair of words $w,w'$, let  $p=w\wedge w'$ and let  $u$, $u'$ be the unique pair of words st.   $w=pu$ and $w'=pu'$.
By definition   we have  $u\wedge u'=\varepsilon$ and  $d_{\rm P}(w,w')=|w|+|w'|-2|p|=|u|+|u'|$, therefore
the following property comes from the definition of ${\cal P}_k$ (see  Fig. \ref{Figure-prefix-distance}):
\begin{fait}
\label{claim0}
With the preceding notation, each of  the following properties holds:

{\rm (i)} $(w,w')\in{\cal P}_k$ is equivalent to $0\le |u|+|u'|\le k$.

{\rm (ii)} $(w,w')\in\uline{{\cal P}_k}$ is equivalent to $1\le |u|+|u'|\le k$.
\end{fait}
\begin{figure}[H]
\begin{center}
\includegraphics[width=8cm,height=4cm]{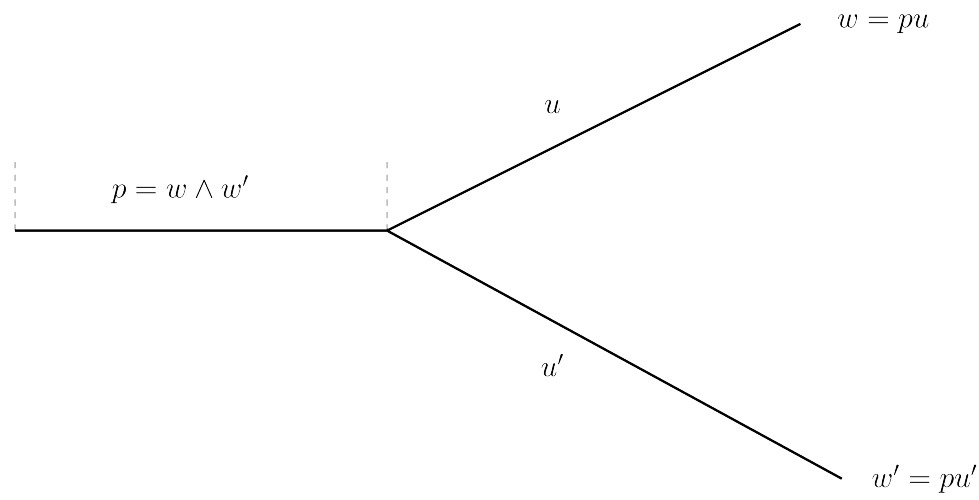}
\end{center}
\caption[]{We have  $(w,w')\in{\cal P}_k$ iff.  $|u|+|u'|\le k$. 
\label{Figure-prefix-distance}
}
\end{figure}
\noindent
In order to prove the further lemma \ref{Pi20},  it is convenient to temporarily move into the more general framework of the factor metric: 
\begin{lem}
\label{Fi20}
Given a positive integer pair $k,k'$ we have ${\cal F}_{k}\cdot {\cal F}_{k'}={\cal F}_{k+k'}$.
\end{lem}
\begin{proof}
We assume wlog. $k\le k'$.
Firstly, we consider a pair of words $(w,w')\in {\cal F}_{k}\cdot {\cal F}_{k'}$. 
By definition, some word $w''\in A^*$ exists st.  we have  $(w,w'')\in {\cal F}_k$, $(w'',w')\in {\cal F}_{k'}$ that is, 
  $d_{\rm F}(w,w'')\le k$ and $d_{\rm F}(w'',w')\le k'$. This implies 
$d_{\rm F}(w,w')\le  d_{\rm F}(w,w'')+d_{\rm F}(w'',w')\le k+k'$, thus $(w,w')\in {\cal F}_{k+k'}$.
Conversely  let  $(w,w')\in { \cal F}_{k+k'}$, and let $f$ be  a word with maximum length  in ${\rm F}(w)\cap {\rm F}(w')$.
Regarding  $|w|-|f|$ and $|w'|-|f|$, exactly one of the two following cases occurs:

\smallskip
\noindent
 \hspace*{2mm} (a) {\it The case where at least one of the integers $|w|-|f|$, $|w'|-|f|$ belongs to $[k,k+k']$}.
Since ${\rm \cal F}_{k+k'}$ is a symmetric relation, wlog. we assume $k\le |w|-|f|\le k+k$'. With this condition
a pair of words $u$, $v$ exist st.  $w\in A^hufvA^{h'}$, with $h+h'=k$.
By construction we have $|ufv|=|w|-k$ and $d_{\rm F}(w,ufv)=k$, thus $(w,ufv)\in{\cal F}_k$. 
We observe that $t_1ft_2\in  {\rm F}(ufv)\cap {\rm F}(w')$ implies $t_1ft_2\in  {\rm F}(w)\cap {\rm F}(w')$. Since $f$ is a maximum length word in ${\rm F}(w)\cap {\rm F}(w')$ we obtain $t_1=t_2=\varepsilon$.
In other words, $f$ remains a maximum length word in  ${\rm F}(ufv)\cap{\rm F}(w')$.
This implies  $d_{\rm F}(ufv,w')=|ufv|+|w'|-2|f|=(|w|-k)+|w'|-2|f|=(|w|+|w'|-2|f|)-k=d_{\rm F}(w,w')-k$. It follows from $d_{\rm F}(w,w')\le k+k'$ that
$d_{\rm F}(ufv,w')\le k'$, thus $(ufv,w')\in{\rm \cal F}_{k'}$. Since we have $(w,ufv)\in{\cal F}_k$, this implies  $(w,w')\in{\cal F}_k\cdot {\cal F}_{k'}$.

\smallskip
\noindent
 \hspace*{2mm} (b) {\it The case where we have   $0\le|w|-|f|< k$ and $0\le |w'|-|f|<k$}.  Since we assume $k\le k'$, by definition we have  $d_{\rm F}(w,f)<k$, $d_{\rm F}(f,w')<k\le k'$
that is,  $(w,f)\in {\cal F}_k$, $(f,w')\in {\cal F}_{k'}$, thus
$(w,w')\in {\cal F}_{k}\cdot {\cal F}_{k'}$.%
\end{proof}

\noindent  Lemma \ref{Fi20} will be further applied  in Sect. \ref{Phi}; regarding ${\cal P}_k$,  it leads to the following statement:
\begin{lem}
\label{Pi20}
Given a positive integer pair $k,k'$ we have ${\cal P}_{k}\cdot {\cal P}_{k'}={\cal P}_{k+k'}$.
\end{lem}
\begin{proof}
Firstly, in the proof of Lemma \ref{Fi20}, by merely substitute $d_{\rm P}$ to $d_{\rm F}$,  the condition $(w,w')\in {\cal P}_{k}\cdot {\cal P}_{k'}$ implies $(w,w')\in {\cal P}_{k+k'}$.
Conversely, take for $f$ a maximum length word in ${\rm P}(w)\cap {\rm P}(w')$.
Once more by substituting $d_{\rm P}$ to $d_{\rm F}$ exactly one of the two following conds. can occur:

\smallskip
\noindent
 \hspace*{2mm} (a) 
In the case where $|w|-|f|\in [k,k+k']$ holds, we have $w\in A^hufvA^{h'}$, with $h=0$, $h'=k$, and $u=\varepsilon$.
We obtain  $(w,fv)\in{\cal P}_k$ and $(fv,w')\in{\rm \cal P}_{k'}$, thus  $(w,w')\in{\cal P}_k\cdot {\cal P}_{k'}$.

\smallskip
\noindent
 \hspace*{2mm} (b) 
In the case where both the conditions $0\le|w|-|f|< k$ and $0\le |w'|-|f|<k$ hold, we directly obtain
 $(w,f)\in {\cal P}_k$, $(f,w')\in {\cal P}_{k'}$, thus $(w,w')\in{\cal P}_k\cdot {\cal P}_{k'}$.
\end{proof}

\noindent Consequently the equations  ${\cal F}_k\cdot{\cal F}_{k'}={\cal F}_{k'}\cdot{\cal F}_{k}$ and ${\cal P}_k\cdot{\cal P}_{k'}={\cal P}_{k'}\cdot{\cal P}_{k}$ hold.
In addition, the following property holds: 
\begin{coro}
\label{Pi2}
Given a pair of positive integer $k,n$ we have ${\cal P}_{nk}={\cal P}_k^n$.
\end{coro}
\begin{proof}
We argue by induction on $n\ge1$. 
Trivially, the equation holds for $n=1$. Assuming that ${\cal P}_{nk}={\cal P}_k^n$ holds,
according to Lemma \ref{Pi20}  we obtain ${\cal P}_{(n+1)k}={\cal P}_{nk}\cdot {\cal P}_{k}={\cal P}_k^n\cdot {\cal P}_{k}={\cal P}_k^{n+1}$.
\end{proof}
\noindent As another consequence of Lemma \ref{Pi20}, in the framework of the prefix metric, Lemma \ref{equiv-c2} leads to the following result:
\begin{lem}
\label{equiv-c2-P}
Given a set $X\subseteq A^*$ and   $k\ge 1$ the three following conds. are equivalent:

{\rm (i)} $X$ satisfies Cond. \ref{2} wrt.  ${\cal P}_k$.

{\rm (ii)}  For every $x\in X$ we have ${\cal P}_{2k}\left(x\right)\cap X=\{x\}$. 

{\rm (iii)} $X$ satisfies Cond. \ref{1} wrt.  ${\cal P}_{2k}$.
\end{lem}
\begin{proof}
Since the relation ${\cal P}_k$ is symmetric we have ${\cal P}_k\cdot{\cal P}_k^{-1}={\cal P}_{k}^2$.
By taking $n=2$ in the statement of Corollary \ref{Pi2}, we obtain ${\cal P}_k\cdot{\cal P}_k^{-1} ={\cal P}_{2k}$. 
The rest of the proof merely consists in substituting ${\cal P}_k$ to $\tau_{d,k}$ in the proof of Lemma \ref{equiv-c2}.
\end{proof}
\subsection{On the regularity of $\uline{{\cal P}_k}$}
In view of Claim \ref{claim0} we introduce the three following sets:
$E$ stands for the set of  all the pairs of non-empty words $(u,u')$, with different initial characters, and st. 
$|u|+|u'|\le k$. In addition $F$ (resp., $G$) stands for the set of all the pairs $(u,\varepsilon)$ (resp., $(\varepsilon,u')$), with $1\le |u|\le k$ (resp., $1\le |u'|\le k$).
By construction, $E$, $F$, and $G$ are finite sets. 
Regarding   Conds. \ref{1}, \ref{2}, the following property will have noticeable  involvement:
\begin{proposit}
\label{Pi-regular}
For every $k\ge 1$, both the relations ${\cal P}_k$ and $\uline{{\cal P}_k}$ are regular.
\end{proposit}
\begin{proof}
-- In what follows we indicate the construction of a finite $A^*\times A^*$-automaton with behavior $\uline{{\cal P}_k}$, namely ${\cal R}_{{\rm P},k}$.
The states  are $0,1,2,3$, the unique initial one being  $0$, and the terminal  being $1,2,3$. The transitions are listed hereunder (see Fig \ref{Automaton-antiref-prefix-2}):

\noindent\hspace*{8mm} $0\xrightarrow{(a,a)}0$, for every $a\in A$;

\noindent\hspace*{8mm} $0\xrightarrow{(u,u')}2$ for every $(u,u')\in E$;

\noindent\hspace*{8mm} $0\xrightarrow{(u,\varepsilon)}1$ for every $(u,\varepsilon)\in F$;

\noindent\hspace*{8mm} $0\xrightarrow{(\varepsilon,u')}3$ for every $(\varepsilon,u')\in G$.

\smallbreak
\noindent\hspace*{2mm}-- Let $(w,w')\in \left|{\cal R}_{{\rm P},k}\right|$, and let $p=w\wedge w'$.
By construction, there are  $t\in \{1,2,3\}$ and $(u,u')\in E\cup F\cup G$,
 st.  $w=pu$, $w'=pu'$, with  $0\xrightarrow{(p,p)}0\xrightarrow {(u,u')}t$ being a  successful path. Since $(u,u')\in E\cup F\cup G$ implies $1\le |u|+|u'|\le k$, we are in the cond. (ii) of 
Claim \ref{claim0}, hence   $(w,w')\in\uline{{\cal P}_k}$ holds.

\smallbreak
\noindent\hspace*{2mm}-- Conversely, assume $(w,w')\in\uline{{\cal P}_k}$. According to Claim \ref{claim0}, words $p$, $u$, and $u'$ exist st.   $p=w\wedge w'$,  $w=pu$, $w'=pu'$, $u\wedge u'=\varepsilon$, and $1\le |u|+|u'|\le k$.
The last cond. implies that at least one of the two words $u,u'$ is non-empty. 
The cond. $u\ne\varepsilon$ with $u'\ne\varepsilon$ implies $(u,u')\in E$. 
Similarly, $u\ne\varepsilon$ and $u'=\varepsilon$ (resp., $u=\varepsilon$ and $u'\ne\varepsilon$), implies $(u,u')\in F$ (resp., $(u,u')\in G$). 
In any case we obtain $(u,u')\in E\cup F\cup G$, therefore by construction some $t\in\{1,2,3\}$ exists st. $0\xrightarrow{(p,p)}0\xrightarrow {(u,u')} t$ is a successful path in ${\cal R}_{{\rm P},k}$.
In other words we have  $(w,w')=(pu,pu')\in \left|{\cal R}_{{\rm P},k}\right|$.

\smallbreak
\noindent\hspace*{2mm}-- 
As a consequence,  we have  $\uline{{\cal P}_k}=\left|{\cal R}_{{\rm P},k}\right|$:  according to Theorem \ref{Elgot-Mezei-0}, the relation  $\uline{{\cal P}_k}$ is regular. 
In addition, according to Prop. \ref{properties-reg-A} the relation  ${\cal P}_k=\uline{{\cal P}_k}\cup id_{A^*}$ itself is  regular.
\end{proof}
\noindent
\begin{example}
{\rm
Let $k=2$. We have $E=\{(a,b): a,b\in A,a\ne b\}$, $F=\{(a,\varepsilon): a\in A\}\cup \{(ab,\varepsilon): a,b\in A\}$, and $G= \{(\varepsilon,a): a\in A\}\cup \{(\varepsilon,ab): a,b\in A\}$ (see Fig. \ref{Automaton-antiref-prefix-2}). 
}
\end{example}
\begin{figure}[H]
\begin{center}
\includegraphics[width=4.5cm,height=4.5cm]{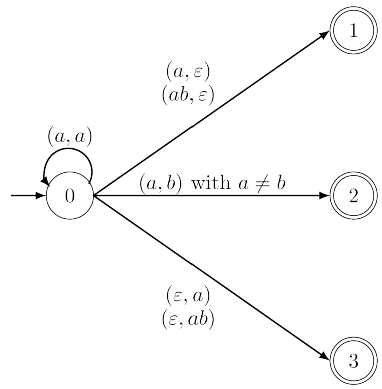}
\end{center}
\caption[]{\small The case where we have $k=2$: in the automaton ${\cal R}_{{\rm P}, k}$, the  edges are muti-labelled ($a$, $b$ stand for every pair of characters in $A$).}
\label{Automaton-antiref-prefix-2}
\end{figure}
\noindent We note that in \cite{Ng16} the author introduces an interesting $(A^*\times A^*)\times {\mathbb N}$-automaton (equivalently transducer with input in $A^*\times A^*$ and output in ${\mathbb N}$) (see Fig. \ref{Figure1111}).
This automaton allows to compute $d_{\rm P}$ as follows:
 for every $(w,w')\in A^*\times A^*$, 
the distance $d_{\rm P}(w,w')$ is the least $d\in{\mathbb N}$ for which $\left((w,w'),d\right)$ is the label of  some successful path.
Furthermore, an alternative proof of the regularity of ${\cal P}_k$ can be obtained.
Indeed, by construction, denoting by ${\cal D}\subseteq (A^*\times A^*)\times {\mathbb N}$ the behavior of such an automaton, we have ${\cal P}_k={\cal D}^{-1}\left([1,k]\right)$. 
Since $[1,k]$ is a finite subset of the one-generator monoid ${\mathbb N}$,
it is regular. 
In addition, since regular relations are closed under inverse ${\cal P}_k$ itself is regular.
 However,  we note that such a construction cannot involve the relation $\uline{{\cal P}_k}$ itself that is, it does not affect  Cond. \ref{1}.
\begin{figure}[H]
\begin{center}
\bigbreak\noindent
\includegraphics[width=6cm,height=2.25cm]{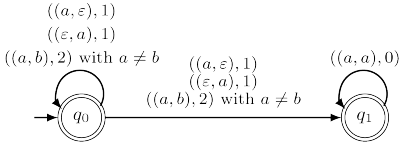}
\end{center}
\caption[]{The automaton from \cite{Ng16} in the case where we have $A=\{a,b\}$}
\label{Figure1111}
\end{figure}
\noindent
A a consequence of Prop. \ref{Pi-regular}, we obtain the following noticeable result:
\begin{proposit}
\label{decid-Pi1}
Given  a regular code $X\subseteq A^*$ and $k\ge 1$,
it can be decided whether $X$ satisfies any of Conds. \ref{1}, \ref{2}, and  \ref{4} wrt.   ${\rm\cal P}_k$.
\end{proposit}
\begin{proof} Let $X$ be a regular code. We consider one by one our conditions:

\smallskip
\noindent\hspace*{2mm}-- {\it Cond. \ref{1}}
According to Prop. \ref{Pi-regular},  the relation $\uline{{\rm \cal P}_k}$ is regular;
according to  Prop. \ref{property-rec1},  $\uline{{\rm \cal P}_k}(X)$ is a regular subset of $A^*$ therefore, according to  Prop. \ref{properties-reg-A}, the set $\uline{{\rm \cal P}_k}(X)\cap X$ itself is  regular.
Consequently, according to Prop. \ref{decidable-bases}, one can decide whether $\uline{{\rm \cal P}_k}(X)\cap X\ne\emptyset$.

\smallskip
\noindent\hspace*{2mm}-- {\it Cond. \ref{2}:}
According to  Lemma \ref{equiv-c2-P} the set $X$ satisfies that cond. wrt.  ${\cal P}_k$ iff. it  satisfies Cond. \ref{1} wrt.  ${\cal P}_{2k}$: 
in view of the above, this can be decided.

\smallskip
\noindent\hspace*{2mm}-- {\it Cond. \ref{4}:}
Since $X$ is regular, according to  Props.  \ref{property-rec1}, \ref{Pi-regular},  the set   $\widehat{{\cal P}_k}(X)={\cal P}_k(X)$ is regular.
According to Prop. \ref{decidable-bases}, one can decide whether this set is a code.
\end{proof}
\noindent
\subsection{Maximal ${\cal P}_k$-independent codes}
Regarding Cond. \ref{3}, we start with the following result:
\begin{proposit}
\label{Comp-Pref-Independent}
Every regular $\uline{{\rm \cal P}_k}$-independent  code can be embedded into some complete one.
\end{proposit}
\begin{proof}
Let $X$ be a regular $\uline{{\rm \cal P}_k}$-independent  code. The result is trivial if $X$ is complete: in the sequel we assume $X$ being non-complete.
By definition, a word $z_0$ exists in   $A^*\setminus {\rm F}(X^*)$. 
Without loss of generality, we assume $|z_0|\ge k$: otherwise, we substitute to $z_0$ some word $z_0u$, with $|u|=k-|z_0|$ (it follows from $z_0\in {\rm F}(z_0u)$ that  $z_0u\notin{\rm F}(X^*)$). 
 Let $a$ be the initial character of $z_0$,
let $b$ be a character different of $a$, and let $z=z_0ab^{|z_0|}$.  
According to Lemma \ref {overlapping-free constr},  $z$ is overlapping-free.
%
We introduce the three following sets:
$U=A^*\setminus (X^*\cup A^*zA^*)$, $Y=z(Uz)^*$, and $Z=X\cup Y$.
According to Theorem \ref{EhRz},  $Z$ is a regular  complete code.
We will prove that $Z$ is $\uline{{\rm \cal P}_k}$-independent that is, $\uline{{\rm \cal P}_k}(X\cup Y)\cap (X\cup Y)=\emptyset$, thus $\uline{{\rm \cal P}_k}(X)\cap X=\uline{{\rm \cal P}_k}(X)\cap Y=\uline{{\rm \cal P}_k}(Y)\cap X=\uline{{\rm \cal P}_k}(Y)\cap Y=\emptyset$. 
Actually, since  $X$ is $\uline{{\rm \cal P}_k}$-independent, we already have $\uline{{\rm \cal P}_k}(X)\cap X=\emptyset$.

\smallskip
\noindent\hspace*{4mm}(a) Firstly, we  prove  that $\uline{{\rm \cal P}_k}(X)\cap Y=\emptyset$. By contradiction  assume that a pair of words $x\in X$ and $y\in Y$ exist st.  $(x,y)\in\uline{{\cal P}_k}$.
By construction we have $|z_0|\ge k$. According to Claim \ref{claim0} we obtain $1\le\left|(x\wedge y)^{-1}x\right|+\left|(x\wedge y)^{-1}y\right|\le k\le |z_0|$. This implies  $\left|(x\wedge y)^{-1}y\right|\le |z_0|$
, thus $\left|x\wedge y\right|\ge |y|-|z_0|$.
By construction, $|y|\ge|z|$ holds:
we obtain $\left|x\wedge y\right|\ge |z_0ab^{|z_0|}|-|z_0|= |ab^{|z_0|}|$, thus $\left|x\wedge y\right|\ge |z_0|+1$. Since  both the words $x\wedge y$ and $z_0$ are prefixes of $y$, this implies $z_0\in {\rm P}(x\wedge y)$, thus  $z_0\in {\rm P}(x)$: 
a contradiction with  $z_0\notin {\rm F}(X^*)$. 

\smallskip
\noindent\hspace*{4mm}(b) Now, by contradiction we assume $\uline{{\rm \cal P}_k}(Y)\cap X\ne\emptyset$. Let $y\in Y$ and $x\in X$ st.  $(y,x)\in \uline{{\rm \cal P}_k}$. 
Since  ${\cal P}_k$ and $\overline{id_{A^*}}$ are symmetrical relations,  $\uline{{\cal P}_k}={\cal P}_k\cap \overline{id_{A^*}}$  itself is symmetrical.
We obtain $(x,y)\in  \uline{{\rm \cal P}_k}$, thus $\uline{{\rm \cal P}_k}(X)\cap Y\ne\emptyset$: this contradicts the conclusion of the preceding case  (a).

\smallskip
\noindent\hspace*{3mm} (c) It remains to prove that $\uline{{\rm \cal P}_k}(Y)\cap Y=\emptyset$. 
Once more arguing  by contradiction, we assume that $y,y'\in Y$ exist st.  $(y,y')\in\uline{{\cal P}_k}$. Let $p=y\wedge y'$, $u=p^{-1}y$, and  $u'=p^{-1}y'$. We compare the words $p$, $y$, and $y'$:

\smallskip
\noindent\hspace*{8mm} (c1)
At first, we assume  $p\in\{y,y'\}$ that is, 
wlog. $y'=p$, thus $y=y'u$. More precisely, it follows from $(y,y')\in \uline{{\cal P}_k}$  that $y\ne y'$, hence $y'$ is a proper prefix of $y$. 
By construction we have  $y,y'\in z(Uz)^*$ that is,  two sequences of words in $U$, namely $u_1,\cdots u_m$ and $u'_1,\cdots u'_n$, exist st.  $y=zu_1zu_2\cdots u_mz$ and $y'=zu'_1z\cdots u'_nz$. It follows from $y=y'u$ that
$zu_1zu_2\cdots u_mz=zu'_1z\cdots u'_nzu$, thus $u_1zu_2\cdots u_mz=u'_1z\cdots u'_nzu$.
With this condition, at least one of the two words $u_1z$ and $u'_1z$ is a prefix of the other one. Since $z$ is overlapping free,
 according to Lemma \ref{prefix-Uz} this implies  $u_1z=u'_1z$, thus $u_1=u'_1$.  
By induction
each of  the equations   $u_2=u'_2$, \dots, and  $u_n=u'_n$ also holds. 
 From the fact that $y'$ is a proper prefix of $y$, we obtain $m\ge n+1$ and
  $u=u_{n+1}z\cdots u_mz$. This implies $|u|\ge |z|\ge|z_0|+1$, thus $|u|\ge k+1$:  
a contradiction with  $|u|=d_{\rm P}(y,y')\le k$.

\smallskip
\noindent\hspace*{8mm}(c2)
Consequently we have $p\notin\{y,y'\}$, thus $u\ne\varepsilon$ and $u'\ne\varepsilon$. 
By construction, $b^{|z_0|}$ is a suffix of $z$, which itself if a suffix of both the words $y,y'\in Y$. 
Consequently, at least one of the words $u$ and $b^{|z_0|}$  is a suffix of the other one;
similarly at least one of the two words $u'$ and $b^{|z_0|}$ is a suffix of the other one. 
It follows from $1\le |u|\le k\le |z_0|$ and   $1\le |u'|\le k\le |z_0|$ that $u,u'$ are non-empty suffixes of $b^{|z_0|}$ that is,
 $u, u'\in bb^*$. But this cond. implies that the word  $pb$ remains  a prefix of both the words $y$ and $y'$: 
a contradiction with  $p=y\wedge y'$  (see Fig. \ref{FigurePrefixProposition3b2}).
\end{proof}
\begin{figure}[H]
\begin{center}
\includegraphics[width=8cm,height=6cm]{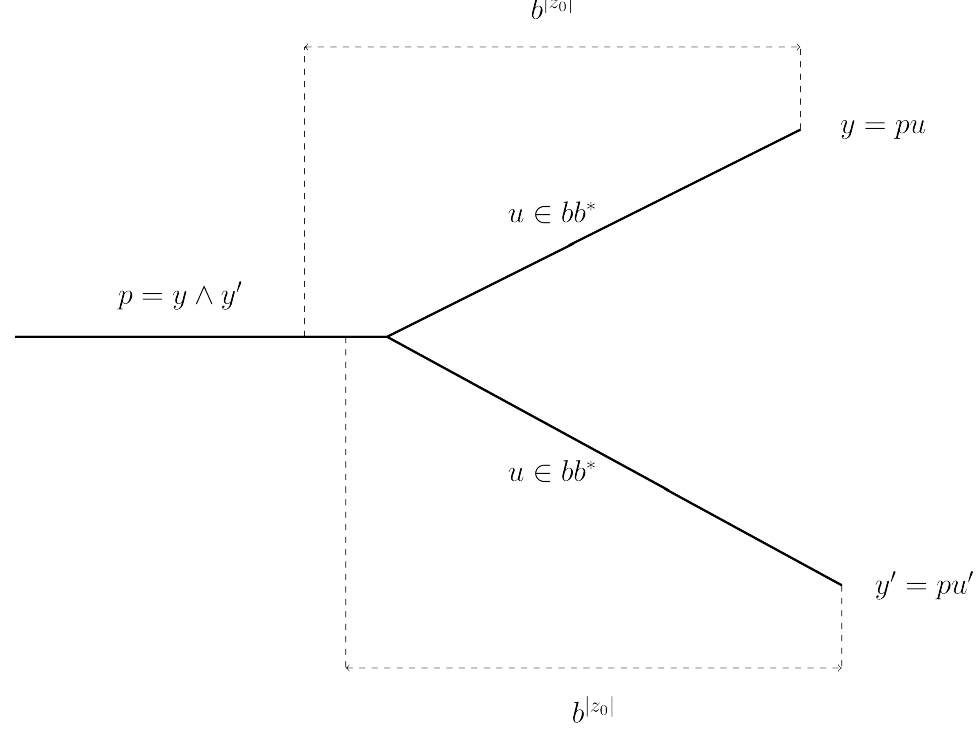}
\end{center}
\caption[]{\small Proof of Prop. \ref{Comp-Pref-Independent}: the case where we have $y,y'\in Y$ and $(y,y')\in{\cal P}_k$, with $p\notin\{y,y'\}$.}
\label{FigurePrefixProposition3b2}
\end{figure}
\noindent As a consequence, we obtain the following result: 
\begin{proposit}
\label{classic11-Pi-k}
Given a regular $\uline{{\rm \cal P}_k}$-independent code $X\subseteq A^*$, the four following conds. are equivalent:

{\rm (i)} $X$  is maximal in the family of $\uline{{\rm \cal P}_k}$-independent codes.

{\rm (ii)} $ X$ is a maximal code.

{\rm (iii)} $X$ is complete.

{\rm (iv)} We have $\mu(X)=1$.
\end{proposit}
\begin{proof}
According to Theorem \ref{classic} for every regular code, Conds.  (ii), (iii), and (iv) are equivalent.
Trivially Cond. (ii) implies Cond. (i).
We prove that Cond. (i) implies Cond. (iii) in arguing by contrapositive.
Assuming $X$ non-complete, 
according to Prop.  \ref{Comp-Pref-Independent}, a regular $\uline{{\rm \cal P}_k}$-independent code  strictly containing $X$ exists, hence $X$ is not maximal as a $\uline{{\rm \cal P}_k}$-independent code.
\end{proof}
\begin{coro}
\label{Comp-ind-max}
Every non-maximal regular $\uline{{\cal P}_k}$-independent code can be embedded into some maximal one.
\end{coro}
\begin{proof}
Let $X$ be  a regular non-maximal regular $\uline{{\cal P}_k}$-independent code. According to Prop. 
\ref{classic11-Pi-k}, $X$ is non-complete. According to Prop.  \ref{Comp-Pref-Independent},  
a complete regular $\uline{{\cal P}_k}$-independent code $Y$ exists  st.  $X\subsetneq Y$. Once more according to Prop. \ref{classic11-Pi-k}, $Y$ is maximal as a regular $\uline{{\cal P}_k}$-independent code.
\end{proof}
\noindent
Finally, we obtain the following result:
\begin{proposit}
\label{c3_Pi}
 One can decide whether a given regular code $X\subseteq A^*$  satisfies  Cond. \ref{3}  wrt.  ${\rm \cal P}_k$. 
\end{proposit}
\begin{proof}
Once more according to Prop. \ref{classic11-Pi-k},  the code $X$ satisfies Cond. \ref{3} iff. $\mu(X)=1$ holds.
According to Prop. \ref{decidable-bases}, one can decide whether $X$ satisfies the last condition.
\end{proof}

\subsection{The suffix metric and the relation ${\cal S}_k$}
Given a pair of words $w,w'$, their {\it suffix} distance is $d_{\rm S}= |w|+|w'|-2|s|$,
where  $s$ denotes the longest word in ${\rm S}(w)\cap{\rm S}(w')$:
set ${\cal S}_{k}=\tau_{d_{\rm S},k}$. For every pair $w,w'\in A^*$, we have $d_{\rm S}(w,w')=d_{\rm P}(w^R,w'^R)$, hence 
$(w,w')\in {\cal S}_k$ is equivalent to $(w^R,w'^R)\in {\cal P}_k$. 
In particular, starting from the preceding automaton ${\cal R}_{{\rm P},k}$ (see proof of Prop. 
\ref{Pi-regular}), 
an automaton  with behavior $\left|\uline{{\cal S}_k}\right|$ can be constructed, whence $\uline{{\cal S}_k}$ is a regular relation (see Fig. \ref{Automaton-antiref-suffix}).
\begin{figure}[H]
\begin{center}
\includegraphics[width=5.104cm,height=5.224cm]{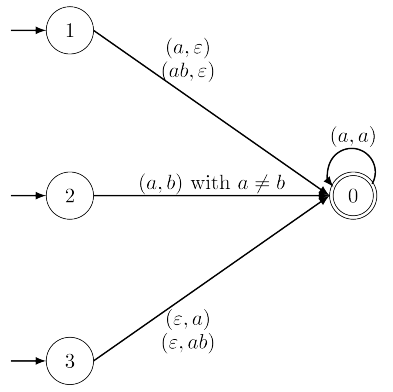}
\end{center}
\caption[]{\small An automaton with behavior $\uline{{\cal S}_2}$.}
\label{Automaton-antiref-suffix}
\end{figure}
\noindent By applying the preceding study to the reversal words, we obtain the following result:
\begin{proposit}
\label{decision-wrt-S}
 Given a regular code, one can decide whether it satisfies any of Conds.  \ref{1}--\ref{4} wrt.  ${\cal S}_k$. 
\end{proposit}
\section{On  the combinatorial structure of  $\uline{{\cal F}_k}$}
\label{Phi}
As in the preceding section, we start with some examples:
\begin{example}
{\rm
Over $A=\{a,b\}$, consider the code $X=\{a,ba,bb\}$ 

\noindent\hspace*{2mm}-- $X$ cannot satisfy Cond. \ref{1} wrt.  ${\cal F}_1$:
 indeed, we have $ba\in\uline{{\cal F}_1}(a)\cap X$.

\noindent\hspace*{2mm}-- Since we have 
 ${\cal F}_1\cdot{\cal F}_1^{-1}\left(a\right)\cap X=\{a,ba\}$, $X$ cannot satisfy the cond. (ii) in Lemma \ref{equiv-c2}, therefore it cannot satisfy Cond. \ref{2}.

\noindent\hspace*{2mm}-- We have $\mu(X)=\mu(a)+\mu(b)\mu(a)+\mu(b)^2=1/2+1/4+1/4=1$, therefore $X$ is a maximal code.
However, since it is not a $\uline{{\cal F}_1}$-independent codes, $X$ cannot satisfy Cond. \ref{3} wrt.  ${\cal F}_1$.

\noindent\hspace*{2mm}-- Finally, it follows from $\varepsilon\in {\cal F}_1(X)$ that  $\widehat {{\cal F}_1}(X)=X\cup {\cal F}_1(X)$ is not a code, whence it cannot satisfy Cond. \ref{4}.
}
\end{example}
\begin{example}
{\rm 
Take $A=\{a,b\}$ and consider the context-free bifix code $X=\{a^nb^n: n\ge 1\}$.

\noindent\hspace*{2mm} --
It follows from $\uline{{\cal F}_1}(X)=\bigcup_{n\ge 1}\{a^{n-1}b^n, a^{n+1}b^n, ba^nb^n, a^nb^{n-1}, a^nb^{n}a, a^nb^{n+1}\}$ that $\uline{{\cal F}_1}(X)\cap X=\emptyset$, thus $X$ is $1$-error-detecting wrt.  ${\cal F}_1$ (Cond. \ref{1}). 

\noindent\hspace*{2mm}--
Regarding error correction, we have $a^{n+1}b^{n+1}\in {\cal F}_1(a^nb^{n+1})\subseteq {\cal F}_1^2(a^nb^n)$,
thus   $a^{n+1}b^{n+1}\in\uline{{\cal F}_1^2}(a^nb^n)$.
This implies $\uline{{\cal F}_1\cdot {\cal F}_1^{-1}}(a^nb^n)\cap X\ne \{a^nb^n\}$, whence $X$ cannot satisfies the cond. (ii) in  Lemma \ref{equiv-c2}
that is, $X$ cannot satisfy Cond. \ref{2} wrt.  ${\cal F}_1$. 

\noindent\hspace*{2mm}-- We have $\mu(X)=\sum_{n\ge 1} \left(\frac{1}{4}\right)^n=1/3<1$, whence $X$ does not satisfy Cond. \ref{3}.

\noindent\hspace*{2mm}-- Finally, since we have $(a^nb^{n-1})(ba^nb^n)=(a^nb^n)(a^nb^n)$, the set $\widehat{{\cal F}_1}(X)={\cal F}_1(X)$  no more satisfies Cond. \ref{4}.
}
\end{example}
\subsection{A few generalities about  ${\cal F}_k$}
Given $w,w'\in A^*$, let $f$ be a maximum length word in  ${\rm F}(w)\cap {\rm F}(w')$ and let $(u,v,u',v')$ be a
 tuple of words st.  $w=ufv=u'fv'$: we have $d_{\rm F}(w,w')=|w|+|w'|-2|f|=|u|+|v|+|u'|+|v'|$.
The following statement, which comes from the definition, provides an extension of Claim \ref{claim0} in the framework of the factor metric 
(see Fig. \ref{config-factor-distance}):
\begin{fait}
\label{claim00}
With the preceding notation, each of the following properties holds:

{\rm (i)} $(w,w')\in {\cal F}_k$ is equivalent to $0\le |u|+|v|+|u'|+|v'|\le k$.

{\rm (ii)} $(w,w')\in\uline{ {\cal F}_k}$ is equivalent to $1\le |u|+|v|+|u'|+|v'|\le k$.
\end{fait}
\begin{figure}[H]
\begin{center}
\includegraphics[width=10cm,height=4cm]{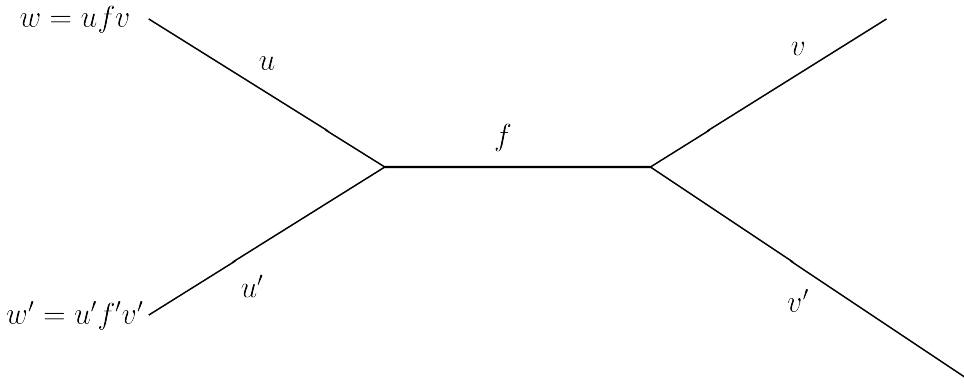}
\end{center}
\caption[]{We have $(w,w')\in{\cal F}_k$ iff.  $|u|+|v|+|u'|+|v'|\le k$.} 
\label{config-factor-distance}
\end{figure}

\noindent As attested by the following example, a maximum length word in ${\rm F}(w)\cap{\rm F}(w')$ needs not to be unique.
Because of this,   the proof of Prop. \ref{Pi-regular}   cannot be extended to the framework  of $\uline{{\cal F}_k}$.
\begin{example}
\label{Facts1}
{\rm  
$w=cccbaba babba b$, $w'=b babba ababaaccca$. There are two words of maximum length in ${\rm F}(w)\cap {\rm F}(w')$, 
namely $f_1=ababa$ and $f_2=babba$. 
}
\end{example}
\noindent 
Since $d_{\rm F}$ is a metric, the relation ${\cal F}_k$ is reflexive and symmetric.
In addition,  according to Lemma \ref{Fi20}, Corollary   \ref{Pi2}, directly translates into the following statement:  
\begin{coro}
\label{Fi2}
Given a pair of positive integer $k,n$ we have ${\cal F}_{nk}={\cal F}_k^n$.
\end{coro}
\noindent 
As a  consequence of Corollary \ref{Fi2}, we obtain the following result:
\begin{lem}
\label{Phi-2}
Given a positive integer $k$ we have ${\cal F}_{k}={\cal F}_{1}^k=({\cal P}_1\cup{\cal S}_1)^k$.
\end{lem}

\begin{proof}
-- With the condition of Corollary \ref{Fi2}, by   taking $n=1$: we obtain ${\cal F}_1^k={\cal F}_{k}$.

\smallskip
\noindent \hspace*{2mm}
-- Trivially, we have ${\cal P}_1\cup{\cal S}_1\subseteq {\cal F}_{1}$.
For the converse, let $(w,w')\in {\cal F}_{1}$, let $f$ be a maximum word in ${\rm F}(w)\cap {\rm F}(w')$, and let $(u,v,u',v')$ st.  $w=ufv$, $w'=u'fv'$.
Acording to Claim \ref{claim00}, we have $0\le |u|+|v|+|u'|+|v'|\le 1$, thus $|u|+|v|+|u'|+|v'|\in\{0,1\}$.
More precisely, at most one of the integers $|u|$, $|v|$, $|u'|$ and $|v'|$ is non-zero: this implies $(w,w')\in{\cal P}_1\cup{\cal S}_1$.
As a consequence we have ${\cal F}_{1}={\cal P}_1\cup{\cal S}_1$, thus 
${\cal F}_{1}^k=({\cal P}_1\cup{\cal S}_1)^k$.
\end{proof}
\noindent
Regarding regular relations, the following result is a consequence, of Lemma \ref{Phi-2}:
\begin{coro}
\label{Fk-F1bar-regular}
 Each of the following properties holds:

{\rm (i)} For every integer $k\ge 1$, the relation ${\cal F}_k$ is regular.

 {\rm  (ii)} The relation $\uline{{\cal F}_1}$ is regular.
\end{coro}
\begin{proof}
According to Prop. \ref{Pi-regular}, ${\cal P}_1$ and ${\cal S}_1$ are regular relations. According to Prop. \ref{properties-reg-A}, the relations ${\cal F}_1={\cal P}_1\cup {\cal S}_1$ and ${\cal F}_k={\cal F}_1^k$ are regular:
this establishes  the property (i). 
Regarding  the property  (ii), we have  $\uline{{\cal F}_1}=\left( {\cal P}_1\cup {\cal S}_1 \right)\cap \overline{id_{A^*}}=\left( {\cal P}_1\cap \overline{id_{A^*}}\right)\cup \left({\cal S}_1\cap \overline{id_{A^*}} \right)=\uline{{\cal P}_1}\cup \uline{{\cal S}_1}$. Once more according to Prop. \ref{Pi-regular}, the relations $\uline{{\cal P}_1}$ and $\uline{{\cal S}_1}$ are regular, whence $\uline{{\cal F}_1}$ itself is regular.
\end{proof}
\noindent 
Regarding Conds. (c1),(c2), the property (ii) in Corollary \ref{Fk-F1bar-regular} may appear promising.
Unfortunately, we do not know whether it could be extended to the relation ${\cal F}_k$,  for any $k\ge 1$.
For instance, taking $k=2$ we have  $\uline{{\cal F}_1}=\left( {\cal P}_1\cup {\cal S}_1 \right)^2\cap \overline{id_{A^*}}=\uline{{\cal P}_1^2}\cup\uline{{\cal P}_1{\cal S}_1}\cup\uline{{\cal S}_1{\cal P}_1}\cup\uline{{\cal S}_1^2}$ thus, according to Corollary \ref{Pi2},
$\uline{{\cal F}_1}=\uline{{\cal P}_2}\cup\uline{{\cal P}_1{\cal S}_1}\cup\uline{{\cal S}_1{\cal P}_1}\cup\uline{{\cal S}_2}$.  
 Although the relations $\uline{{\cal P}_2}$ and $\uline{{\cal S}_2}$ are regular (see  Prop. \ref{Pi-regular}), we do  not know whether the same holds for $\uline{{\cal P}_1{\cal S}_1}$ and $\uline{{\cal S}_1{\cal P}_1}$.
However, a clever examination of  the structure of ${\cal F}_k\subseteq A^*\times A^*$ will allow to  overcome this obstacle: to be more precise, afterwards we will prove that $\uline{{\cal F}_k}$ preserves the regularity of sets.
\subsection{A set covering for  $\uline{{\cal F}_k}$}
\label{structure}
 
\label{automaton F_k}
We start by the following property: regarding Claim \ref{claim00},
it actually allows to get away of the  maximum length  constraint over the words in ${\rm F}(w)\cap {\rm F}(w')$. 
\begin{lem}
\label{charact-d_F}
Given a pair of words $(w,w')\in A^*\times A^*$, we have $(w,w')\in {\cal F}_k$ iff.  a tuple of words $(g, u,u',v,v')$ exists st.  each of the following conds. holds:

{\rm (i)}  $0\le |u|+|u'|+|v|+|v'|\le k$ holds.

{\rm (ii)} We have $w=ugv$ and $w'=u'gv'$.
\end{lem}
\begin{proof}
Let $(w,w')\in {\cal F}_k$,
let $f$ be a maximum length word in ${\rm F}(w)\cap {\rm F}(w')$, and let $(u,u',v,v')$ st. $w=ufv$, $w'=u'fv'$.  
Taking $g=f$, trivially we obtain the cond. (ii) of Lemma \ref{charact-d_F}.
In addition, according to Claim \ref{claim00}, we obtain the cond. (i).
Conversely, assume that there is  a tuple of words $(g, u,u',v,v')$ satisfying  both the conds. (i) and (ii).
With the cond. (ii), the word $g$ belongs to ${\rm F}(w)\cap {\rm F}(w')$:
by the maximality of $|f|$ we have $|g|\le |f|$. This implies
$d_{\rm F}(w,w')=|w|+|w'|-2|f|\le |w|+|w'|-2|g|$, thus  
$ d_{\rm F}(w,w')\le |u|+|u'|+|v|+|v'|$. In  view of the cond. (i) we obtain $d_{\rm F}(w,w')\le k$, thus $(w,w')\in {\cal F}_k$.
\end{proof}
\noindent
Notice that, in the statement of Lemma \ref{charact-d_F}, the word $g$ needs not to be a factor of any maximum length word  in ${\rm F}(w)\cap {\rm F}(w')$, as attested by what follows:
\begin{example} (Example \ref{Facts1} continued)
{\rm  
Let $w= ccc baba babba b$, $w'=b babba ababaa ccc a$, and $k=23$. Recall that there  are two words of maximum length in ${\rm F}(w)\cap {\rm F}(w')$, 
namely $f_1=ababa$ and $f_2=babba$: we have $d_{\rm F}(w,w')=|w|+|w'|-2|f_1|=13+16-10=19$, thus $(w,w')\in{\cal F}_k$. 
Taking $g=ccc$, we have
$w=ugv$, $w'=u'gv$', with $u=\varepsilon$, $v=babababba b$, $u'=bbabbaababaa$, and $v'=a$.
Although we have neither $g\in {\rm F}(f_1)$ nor $g\in  {\rm F}(f_2)$, it follows from  $|u|+|u'|+|v|+|v'|=23\le k$ that  $(g,u,u',v,v')$ satisfies both  Conds. (i), (ii) of Lemma \ref{charact-d_F}.
}
\end{example}
\noindent In view of Lemma \ref{charact-d_F}, in what follows,   
we indicate the construction of a finite family of subsets of $A^*\times A^*$, namely $\left(S_\omega\right)_{\omega\in\Omega_k}$.
Firstly, we denote by $\Omega_k$ the set of the tuple $(u,u',v,v')$ satisfying  the cond. (i) in Lemma \ref{charact-d_F}. By construction, $\Omega_k$ is finite.
Secondly, for each $\omega=(u,u',v,v')\in \Omega_k$ we denote by $S_\omega$ the set of the pairs $(ugv,u'gv')$, for all $g\in A^*$. 
The following result emphasizes a connection between the family $\left(S_\omega\right)_{\omega\in\Omega_k}$ and the relation ${\cal F}_k$:
\begin{lem}
\label{automaton-A-omega}
Each of the following properties holds:

{\rm(i)} For any $\omega\in\Omega_k$, the relation $S_\omega\subseteq A^*\times A^*$  is  regular. 

{\rm (ii)} We have ${\cal F}_k=\bigcup\limits_{\omega\in\Omega_k}S_\omega$. 
\end{lem}
\begin{proof}
 -- For establishing  the property  (i), we observe that, given $\omega=(u,u',v,v')\in\Omega_k$, the relation  $S_\omega$ is the behavior of a finite automaton, namely ${\cal A}_\omega$.
The states are $0,1,2$, the initial state being $0$ and the terminal one being $2$.
The transitions are $0\xrightarrow{(u,u')}1$, $1\xrightarrow{(v,v')}2$, and $1\xrightarrow{(a,a)}1$, for every $a\in A$ (see Fig. \ref{AutomatonFact1}).
By construction the successful paths are
$0\xrightarrow{(u,u')}1\xrightarrow{(a_1,a_1)}1 \cdots 1\xrightarrow{(a_n,a_n)}1\xrightarrow{(v,v')}2$, with $n\ge 0$ and $a_1,\cdots, a_n\in A$.
Since the corresponding labels are $(ugv,u'gv')$, for all $g=a_1\cdots a_n\in A^*$, the behavior  of ${\cal A}_\omega$ is $S_\omega$.
\begin{figure}[H]
\begin{center}
\includegraphics[width=8.5cm,height=1.75cm]{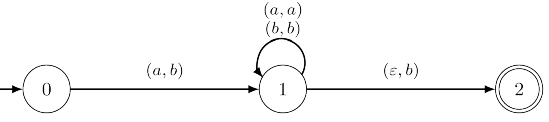}
\end{center}
\caption[]{\small Automaton ${\cal A}_\omega$, with $\omega=(a,b,\varepsilon,b)$.}
\label{AutomatonFact1}
\end{figure}  
\noindent\hspace*{2mm}-- Now we  proceed to establish the property (ii).
Firstly, we prove that ${\cal F}_k\subseteq\bigcup_{\omega\in\Omega_k} S_\omega$. Assuming that $(w,w')\in{\cal F}_k$,
according to  Lemma \ref{charact-d_F}, a tuple of words  $(g,u,u',v,v')$ exists  st.   $0\le |u|+|u'|+|v|+|v'|\le k$, $w=ugv$, and $w'=u'g'v'$: set $\omega=(u,u',v,v')$.
By construction the cond. $0\le |u|+|u'|+|v|+|v'|\le k$ implies $\omega\in \Omega_k$; in addition $w=ugv$, $w'=u'gv'$ implies $(w,w')\in S_\omega$. 
For proving that $\bigcup_{\omega\in\Omega_k}S_\omega\subseteq {\cal F}_k$, we consider  $\omega=(u,u',v,v')\in\Omega_k$, and  $(w,w')\in S_\omega$. By definition,  $g\in  A^*$ exists st.  $w=ugv$ and $w'=u'gv'$.
On the one hand, by construction $\omega\in\Omega_k$ implies $|u|+|u'|+|v|+|v'|\le k$: we obtain the cond. (i) of Lemma \ref{charact-d_F}. 
On the other hand, from the fact that the equations  $w=ugv$ and $w'=u'gv'$ hold, we are further in the cond. (ii) of the same lemma: this implies $(w,w')\in {\cal F}_k$.
\end{proof}
\noindent Lemma \ref{automaton-A-omega} attests that the family $\left(S_\omega\right)_{\omega\in\Omega_k}$ constitutes a set covering for ${\cal F}_k\subseteq A^*\times A^*$.
 In that family,  the sets need not to be  pairwise disjoint, even pairwise different.
For instance,  let $k\ge 2$, and let $u,v\in A^*$ such that $|u|+|v|\le k/2$. 
With this cond. both the tuples $\omega_0=(\varepsilon,\varepsilon,\varepsilon,\varepsilon)$ and $\omega_1=(u,u,v,v)$ belong to  $\Omega_k$. 
From the fact that 
$(\varepsilon\cdot uv\cdot \varepsilon, \varepsilon\cdot uv\cdot \varepsilon)=(u\cdot \varepsilon\cdot v,u\cdot \varepsilon\cdot v)$, 
we obtain $(uv,uv)\in S_{\omega_0}\cap S_{\omega_1}$. 
As another example, taking $\omega_2=(u,u',pt,pt')\in\Omega_k$ and $\omega_3=(u,u',t,t')\in\Omega_k$,  we have $S_{\omega_2}=S_{\omega_3}$.
Indeed, $(w,w')\in\omega_2$ holds iff. some $g\in A^*$ exists st.  $w=u\cdot g\cdot pt$ and $w'=u'\cdot g\cdot pt'$ that is,  $w=u\cdot gp\cdot t$ and $w'=u'\cdot gp\cdot t'$. 
Nevertheless,  such redundancies have no incidence on the rest of our  study. In particular they  do not involve the properties of the set $R_\omega$, which is further constructed.

\medskip\noindent
In the sequel, starting from $(S_\omega)_{\omega\in\Omega_k}$, we construct a covering for $\uline{{\cal F}_k}$. Beforehand, we note that the following property holds:
\begin{fait}
\label{claim20}
Let  $\omega=(u,u',v,v')\in\Omega_k$. Then $u=u'$ with $v=v'$ implies $S_\omega\subseteq id_{A^*}$.
\end{fait}
\begin{proof}
$u=u'$ with  $v=v'$ implies $S_\omega=\{(ugv,ugv): g\in A^*\}$. Furthermore $(w,w')\in S_\omega$ implies $w'=w$.
\end{proof}
\noindent In view of Claim \ref{claim20}, we  introduce the set  $\Omega'_k$ of all the tuple $(u,u',v,v')\in\Omega_k$ st.  at least one of the  conds. $u\ne u'$  or $v\ne v'$ holds. With this notation, the following property holds:
\begin{lem}
\label{claim21}
We have  $\uline{{\cal F}_k}=\bigcup\limits_{\omega\in\Omega'_k}\left(S_\omega\setminus id_{A^*}\right)$.
\end{lem}
\begin{proof}
According to Lemma \ref{automaton-A-omega}, we have $\uline{{\cal F}_k}=\overline{id_{A^*}}\cap\bigcup_{\omega\in\Omega_k}S_\omega$, whence
$\uline{{\cal F}_k}$ is the union of the sets $\bigcup_{\omega\in\Omega'_k} \left(\overline{id_{A^*}}\cap S_\omega\right)=\bigcup_{\omega\in\Omega'_k} \left(S_\omega\setminus id_{A^*}\right)$ and $\bigcup_{(u,u,v,v)\in\Omega_k}\left(S_{(u,u,v,v)}\setminus id_{A^*}\right)$.
Moreover, in view of  Claim   \ref{claim20} the identity $\bigcup_{(u,u,v,v)\in\Omega_k}\left(S_{(u,u,v,v)}\setminus id_{A^*}\right)=\emptyset$ holds.
\end{proof}
\noindent
\subsection{On the combinatorial structure of  $S_\omega\cap id_{A^*}$}
The following result brings a noticeable combinatorial characterization of those tuples $(g,u,u',v,v')$ satisfying the equation $ugv=ugv'$:
\begin{lem}
\label{equivalence w=w'}
Let $\omega=(u,u',v,v')\in\Omega'_k$ and  $g\in A^*$ st.  $(ugv,u'gv')\in S_\omega$. 
We have $ugv=u'gv'$ iff.  $\alpha\in A^*$, $\beta\in A^+$, and $n\in{\mathbb N}$ exist st.
exactly one of the following conds. holds:

{\rm (i)}   $u=u'\alpha\beta$,  $v'=\beta\alpha v$, and $g\in (\alpha\beta)^n\alpha$.

{\rm (ii)}   $u'=u\alpha\beta$,  $v=\beta\alpha v'$, and $g\in  (\alpha\beta)^n\alpha$.
\end{lem}
\begin{proof}

-- With the cond. (i) we have $ugv=(u'\alpha\beta)(\alpha\beta)^n\alpha v=u'(\alpha\beta)^n\alpha(\beta\alpha v)=u'(\alpha\beta)^n\alpha v'=u'gv'$.
Similarly, the cond. (ii) implies  $u'gv'=(u\alpha\beta)(\alpha\beta)^n\alpha v'=u(\alpha\beta)^n\alpha (\beta\alpha v')=ugv$. 

\smallskip
\noindent\hspace*{2mm}--  Conversely, assume that the equation $ugv=u'gv'$ holds.
It follows from  $|u|+|v|=|u'|+|v'|$ that the conds. $u=u'$ and  $v=v'$ are equivalent. 
Furthermore, $(u,u',v,v')\in\Omega'_k$ implies  $u\ne u'$ and $v\ne v'$. 
Consequently, either $u'$ is a proper prefix of $u$  (and $v$ is  a proper suffix of $v'$), or $u$ is a proper prefix of $u'$  (and  $v'$ is a proper suffix of $v$).

\smallskip
\hspace* {2mm} (a)  In the case where  $u'$ is a proper prefix of $u$, let $t,t'\in A^+$ st.  $u=u't$ and $v'=t'v$.
The equation $ugv=u'gv'$ implies $u'tgv=u'gt'v$, thus $tg=gt'$.
According to Prop. \ref{conjugacy}, $t$, $t'$ are conjugate words, furthermore
 $\alpha\in A^*$, $\beta\in A^+$, and $n\in{\mathbb N}$ exist st.  $t=\alpha\beta$, $t'=\beta\alpha$, and $g\in (\alpha\beta)^n\alpha$.
moreover we have  $u=u't=u'\alpha\beta$, $v'=t'v=\beta\alpha v$:
we are in the cond.  (i). 

\smallskip
\hspace* {2mm} (b) In the case where  $u$ is a proper prefix of $u'$ and $v'$ is  a proper suffix of $v$, we note that,
by construction, the conds. $\omega=(u,u',v,v')\in\Omega'_k$ and $\omega'=(u',u,v',v)\in\Omega'_k$ are equivalent. 
 By substituting $(u',u,v',v)$ to $(u,u',v,v')$ in the arguments we applied in  the preceding case (a),
we obtain  the cond. (ii) of our lemma.
\end{proof}
\noindent
Let $\omega=(u,u',v,v')\in \Omega'_k$. Note that Lemma \ref{equivalence w=w'} does not guarantee the unicity of the pair of words $(\alpha,\beta)$. Accordingly, we introduce the following sets:

\smallskip
--   $R^{\rm (i)}_\omega$ is the union of the sets $R^{\rm (i)}_{\omega,\alpha,\beta}=\{u(\alpha\beta)^n\alpha v: n\ge 0\}$, for all  word pairs $\alpha\in A^*$, $\beta\in A^+$ st.  $u=u'\alpha\beta$ and $v'=\beta\alpha v$.

--  $R^{\rm (ii)}_\omega$ is  the union of the sets $R^{\rm (ii)}_{\omega,\alpha,\beta}=\{u'(\alpha\beta)^n\alpha v': n\ge 0\}$, for all word pairs $\alpha\in A^*$, $\beta\in A^+$ st.  $u'=u\alpha\beta$ and $v=\beta\alpha v'$.

-- $R_\omega=R^{\rm (i)}_\omega\cup R^{\rm (ii)}_\omega$.

\smallskip\noindent As indicated above, either $u'$ is a proper prefix of $u$ or $u$ is a proper prefix of $u'$. Accordingly at most one of the sets $R^{\rm (i)}_\omega$,  $R^{\rm (ii)}_\omega$ may be non-empty. 
\begin{example}
{\rm
Let $A=\{a,b\}$, $k=10$.

\noindent\hspace*{2mm} 
-- Firstly, consider the tuple $\omega=(aba,a,b,b^3)\in\Omega'_k$. 
Since $u'$ is a proper prefix of $u$, only the cond. (i) of Lemma \ref{equivalence w=w'} may hold, whence we have   $R^{\rm (ii)}_\omega=\emptyset$.
With the preceding notation, $u=u'\alpha\beta$ implies $\alpha\beta=ba$; similarly $v'=\beta\alpha v$ implies $ \beta\alpha=b^2$.
It is straightforward to verify that no pair of words $(\alpha,\beta)$ may satisfies such constraints, hence we have  $R^{\rm (i)}_\omega=\emptyset$
and   $R_\omega=\emptyset$.
 
\noindent\hspace*{2mm} 
-- Now, consider the tuple $\omega=(aba,a,b^2,ab^3)\in\Omega'_k$. Once more $u'$ is a proper prefix of $u$, thus  $R^{\rm (ii)}_\omega=\emptyset$ holds.
Regarding  the set $R^{\rm (i)}_\omega$ we have  $\alpha\beta=ba$,  $\beta\alpha=ab$.
The set of the pairs  $(\alpha, \beta)$ st. $\alpha\beta=ab$ is $\{(\varepsilon,ab), (a,b)\}$ (recall that we set $\beta\ne\varepsilon$).
Similarly, the equation  $\beta\alpha=ba$ is satisfied by  the pairs  $(\alpha, \beta)\in \{(\varepsilon,ba), (a,b)\}$.
Consequently only the pair $(a,b)$ may satisfy both the preceding constraints.  
 We verify that the equations $ugv= aba\cdot (ba)^n b \cdot b^2= a\cdot (ba)^nb\cdot ab^3=u'gv'$ hold.
This corresponds to 
$R_\omega=R_{\omega,a,b}^{\rm (i)}=\{(ab)^{n+2}b^2: n\ge 0\}$.
}
\end{example}

\smallbreak\noindent
 As direct consequences of Lemma  \ref{equivalence w=w'} the  two following properties holds:
\begin{coro}
\label{equivalence w=w'-R}
Given $\omega\in\Omega'_k$ and $(w,w')\in S_\omega$,
the conds.  $w=w'$ and $w\in R_\omega$ are equivalent. 
\end{coro}
\begin{proof}
By construction   $g\in A^*$ exists st.  $w=ugv$, $w'=u'gv'$. 
We have $ugv=u'gv'$ iff. 
exactly one of the  conds. (i) or (ii) of  Lemma \ref{equivalence w=w'} is satisfied that is,
  $w\in R^{\rm (i)}_\omega\cup R^{\rm (ii)}_\omega=R_\omega$.
\end{proof}
\begin{coro}
\label{coro2}
Given $\omega\in\Omega'_k$ and $X\subseteq A^*$ we have $ \uline{S_\omega}(X)=S_\omega(X\setminus R_\omega)$.
\end{coro}
\begin{proof}
We have $y\in \uline{S_\omega}(X)$  iff.  $x\in X$ exists st.  $(x,y)\in S_\omega$ with  $x\ne y$.  According to Corollary   \ref{equivalence w=w'-R}, $x\ne y$ is equivalent to $x\notin R_\omega$, hence 
$y\in \uline{S_\omega}(X)$ is equivalent to $y\in \uline{S_\omega}(X\setminus R_\omega)$. 
\end{proof}
\noindent
We conclude Sect. \ref{Phi} by proving that ${\cal F}_k$ is a regularity-preserving relation:
\begin{proposit}
\label{Fk(X)-reg}
If $X\subseteq A^*$ is regular, then $\uline{{\cal F}_k}(X)$ is a regular subset of $A^*$.
\end{proposit}
\begin{proof}
We proceed through the three following steps:
 
\smallskip
\noindent\hspace*{2mm} -- Let $\omega=(u,v,u',v')\in\Omega'_k$, and let 
$\alpha$, $\beta$ be words satisfying the cond. of Lemma \ref{equivalence w=w'}. The set  $R^{\rm (i)}_{\omega,\alpha,\beta}=\{u(\alpha\beta)^n\alpha v: n\ge 0\}$ is regular:
indeed, classically it is the behavior of the finite $A^*$-automaton represented in Figure \ref{automaton-alpha-beta}.
\begin{figure}[H]
\begin{center}
\includegraphics[width=7.5cm,height=1.5cm]{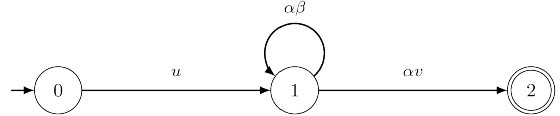}
\end{center}
\caption[]{Proof of Prop. \ref{Fk(X)-reg}. A finite $A^*$-automaton with behavior $\{u(\alpha\beta)^n\alpha v: n\ge 0\}$.} 
\label{automaton-alpha-beta}
\end{figure}
\noindent The set $R^{\rm (i)}_\omega$ is the union of the sets $R^{\rm (i)}_{\omega,\alpha,\beta}$, for all the  pairs of words $(\alpha,\beta)$ satisfying the cond. of Lemma \ref{equivalence w=w'}. 
Since in any case we have $|\alpha|,|\beta|\le \max\{|u|,|u'|,|v|,|v'|\}\le k$, there are only  a finite numbers of such pairs $(\alpha,\beta)$, hence according to Prop. \ref{properties-reg-A}, $R^{\rm (i)}_\omega$ is a regular set.
 Similarly, $R^{\rm (ii)}_\omega$ is a regular, hence $R_\omega=R^{\rm (i)}_\omega\cup R^{\rm (ii)}$ itself is regular. 

\smallbreak
\noindent
\hspace*{2mm}-- Let $X\subseteq A^*$ a regular set. 
According to Corollary  \ref{coro2}, for every $\omega\in\Omega'_k$ we have $\uline{S_\omega}\left(X\right)=S_\omega\left(X\cap\overline{ R_\omega}\right)$.
According to Prop. \ref{properties-reg-A}, since $R_\omega$ is a regular subset of $A^*$, the same holds for the set $\overline{R_\omega}$.
Since the relation $S_\omega\subseteq A^*\times A^*$ is regular (see Lemma \ref{automaton-A-omega}),
according to Prop. \ref{property-rec1}, the set $\uline{S_\omega}(X)= S_\omega\left(X\cap\overline{R_\omega}\right)$ itself is regular.

\smallbreak\noindent
\hspace*{2mm}-- According to   
Lemma \ref{claim21} we have 
$\uline{{\cal F}_k}\left(X\right)=\bigcup_{\omega\in\Omega'_k}\uline{S_\omega}\left(X\right)$.
As established above, for each $\omega\in\Omega'_k$ the set $\uline{S_\omega}\left(X\right)$ is regular.
Since $\Omega'_k$ is  a finite set, 
$\uline{{\cal F}_k}\left(X\right)$ itself is regular.
\end{proof}
\noindent
\section{Error detection conds. wrt.  ${\cal F}_k$}
\label{Decidability-Fk}
We  now have enough properties to establish decidability results.
First and foremost,  according to Corollary \ref{Fi2}, and since ${\cal F}_k$ is a symmetric relation,  Lemma \ref{equiv-c2} translates as follows:
\begin{lem}
\label{equiv-c2-F}
Given a set $X\subseteq A^*$ and $k\ge 1$ the following conds. are equivalent:

{\rm (i)} $X$ satisfies Cond. \ref{2} wrt.  ${\cal F}_k$.

{\rm (ii)}  For every $x\in X$ we have ${\cal F}_{2k}\left(x\right)\cap X=\{x\}$.

{\rm (iii)} $X$ satisfies Cond. \ref{1} wrt.  ${\cal F}_{2k}$.
\end{lem}
\noindent
In view of Prop. \ref{Fk(X)-reg}, 
 the  following decidability result holds:
\begin{proposit}
\label{Phi-regular}
It can be decided whether a given regular set satisfies any of Conds. \ref{1}, \ref{2}, or \ref{4}  wrt.   ${\cal F}_k$.
\end{proposit}
\begin{proof}
Let $X\subseteq A^*$ be a regular set. We examine one by one our conditions:

\smallskip
\noindent\hspace*{2mm}
-- {\it Cond. \ref{1}:}
According to Prop. \ref{Fk(X)-reg},  $\uline{{\cal F}_k}(X)$ is a regular subset of $A^*$, therefore 
 $\uline{{\cal F}_k}(X)\cap X$ itself is regular. According to Prop. \ref{decidable-bases}, one can decide whether it is the empty set or not.

\smallskip
\noindent\hspace*{2mm}
-- {\it Cond. \ref{2}:}
Since ${\cal F}_k$ is  a symmetric relation, according to  Lemma \ref{equiv-c2-F}, $X$ satisfies Cond. \ref{2} iff.  it  satisfies Cond. \ref{1} wrt.  ${\cal F}_{2k}$:
as indicated above,  this cond. can be can decided.

\smallskip
\noindent\hspace*{2mm}
-- {\it Cond. \ref{4}:}
According to Corollary \ref{Fk-F1bar-regular}  ${\cal F}_k$ is a regular relation. By Prop. \ref{property-rec1},  the set $\widehat{{\cal F}_k}(X)={\cal F}_k(X)$ is regular.
According to Prop. \ref{properties-reg-A} $X\cup {\cal F}_k(X)$ is a regular set furthermore,
according to Prop. \ref{decidable-bases},  one can decide whether it is  is a code or not.
\end{proof}
\noindent
\smallbreak
\noindent Regarding  Cond. \ref{3}, we start by proving the following result:
\begin{proposit}
\label{Com-Fact--Independent}
Every  regular $\uline{{\cal F}_k}$-independent code  can be embedded into some  complete one.
\end{proposit}
\begin{proof}
Let $X$ be a regular $\uline{{\cal F}_k}$-independent code: wlog. we assume $X$ non-complete.
At first, we will construct a special word $z_1\in A^*\setminus {\rm F}(X^*)$.
Let  $z_0\notin {\rm F}(X^*)$,  with $|z_0|\ge k$, let $a$ be  its initial character, and let $b$ be a character different of $a$.
Consider the word $z=z_0ab^{|z_0|}$ as constructed in the proof of Prop. \ref{Comp-Pref-Independent}  and
set $z_1=a^{|z|}bz=a^{2|z_0|+1}bz_0ab^{|z_0|}$. According to Prop. \ref{overlapping-free constr}, $z_1$ is overlapping-free, therefore as indicated in the preliminaries,
$z^Rba^{|z|}$, the  reversal of $z_1$, remains  overlapping-free.
Now, we set  $U_1=A^*\setminus \left(X^*\cup A^*z_1A^*\right)$, $Y_1=z_1\left(U_1z_1\right)^*$, and $Z_1=X\cup Y_1$.
According to  Theorem \ref{EhRz}, the set  $Z_1$ is  a regular complete  code. 
In what follows we prove that $Z_1$   is $\uline{{\rm \cal F}_k}$-independent that is, $\uline{{\rm \cal F}_k}(X\cup Y_1)\cap (X\cup Y_1)=\emptyset$.
Since $X$ itself  is $\uline{{\cal F}_k}$-independent, this amounts to prove that each of the three equations
$\uline{{\rm \cal F}_k}(X)\cap Y_1=\emptyset$, $\uline{{\rm \cal F}_k}(Y_1)\cap X=\emptyset$, and $\uline{{\rm \cal F}_k}(Y_1)\cap Y_1=\emptyset$ holds.
\smallskip
\noindent\hspace*{4mm}(a) For proving that $\uline{{\cal F}_k}(X)\cap Y_1=\emptyset$ holds, we argue by contradiction:
let  $x\in X$, $y\in Y_1$ st.   $(x,y)\in\uline{{\rm \cal F}_k}$. 
Let $f$ be a word with maximum length in ${\rm F}(x)\cap {\rm F}(y)$ and let $u,v\in A^*$ st.  $y=ufv$.
According to Claim \ref{claim00}, we have $|u|+|v|\le d_{\rm F}(x,y)\le k$.
By construction  the word $a^{|z_0|}$ (resp., $b^{|z_0|}$)  is a prefix (resp.,  suffix) of $y\in Y_1$,
 therefore at least one of the words $u$ and $a^{|z_0|}$ (resp., $v$ and $b^{|z_0|}$)  is a prefix (resp., suffix) of the other one. 
More precisely, it follows from $|u|\le k\le |z_0|$ and $|v|\le k\le |z_0|$  that we have $u\in {\rm P}(a^{|z_0|})$ and  $v\in{\rm S}(b^{|z_0|})$. 
On the other hand, by construction, $y\in Y_1$ implies either $y=z_1$, or  $y=z_1wz_1$ for some $w\in A^+$.
With the first condition, the equation $f=u^{-1}z_1v^{-1}=a^{2|z_0|+1-|u|}bz_0b^{|z_0|-|v|}$ holds.
With the second cond. we have 
$f=u^{-1}wv^{-1}=a^{2|z_0|+1-|u|}wbz_0b^{|z_0|-|v|}$.
In each case 
 $z_0$ is a factor of $f\in{\rm F}(x)$:
a contradiction with $z_0\notin {\rm F}(X^*)$. Consequently we have $\uline{{\cal F}_k}(X)\cap Y_1=\emptyset$.
\smallskip
\noindent\hspace*{4mm}(b)  By contradiction we assume $\uline{{\rm \cal F}_k}(Y_1)\cap X\ne\emptyset$. Let $y\in Y_1$ and $x\in X$ st.  $(y,x)\in \uline{{\rm \cal F}_k}$. 
Since  the relations ${\cal F}_k$ and $\overline{id_{A^*}}$ are symmetrical,  $\uline{{\cal F}_k}={\cal F}_k\cap \overline{id_{A^*}}$  itself is symmetrical, 
therefore, we have  $(x,y)\in  \uline{{\rm \cal F}_k}$, thus $\uline{{\rm \cal F}_k}(X)\cap Y\ne\emptyset$: this contradicts the conclusion of Case  (a).
\smallskip
\noindent\hspace*{4mm}(c) It remains to prove that $\uline{{\rm \cal F}_k}(Y_1)\cap Y_1=\emptyset$. 
By contradiction, we assume that a pair of different words $y,y'\in Y_1$ exist st.  $(y,y')\in {\cal F}_k$. Let $f$ be  a word with maximum length in ${\rm F}(y)\cap {\rm F}(y')$. 
Once more according to Claim \ref{claim00}, words  $u,u',v,v'$ exist st.  $y=ufv$, $y'=u'fv'$, with $|u|+|u'|+|v|+|v'|=d_{\rm F}(w,w')\le k$. 

\smallskip
\noindent\hspace*{8mm}(c1) At first, we  compare the words $v$, $v'$ with $\varepsilon$.
Firstly, by contradiction assume that both the conds. $v\ne\varepsilon$, $v'\ne\varepsilon$ hold. Necessarily we have  $2\le |v|+|v'|\le k\le |z_0|$.
By construction $v,v'\in {\rm S}(Y_1)$ holds: this implies  $v,v'\in {\rm S}(b^{|z_0|})$, whence
$i,j\in [1,|z_0|]$ exist st.  $v=b^i$, $v'=b^j$ (see Fig. \ref{FigureFactorComp-b1}).
We obtain $fb\in {\rm F}(y)\cap {\rm F}(y')$, a  contradiction with $|f|$ being maximum.
Consequently at least one of the conds. $v=\varepsilon$ or $v'=\varepsilon$ holds: wlog. we assume $v'=\varepsilon$,
thus $f\in{\rm  S}(y')$. 
On the one hand, according to the maximality of $|f|$, it follows from $z_1\in {\rm F}(y)\cap {\rm F}(y')$ that $|f|\ge |z_1|$.
Since we have $f,z_1\in {\rm S}(y')$ we obtain $f\in A^*z_1$, that is, $fv\in A^*z_1v$: in particular we have $|fv|\ge |z_1v|\ge |z_1|$. 
On the other hand, by construction both the words $z_1$ and $fv$ are suffixes of $y$: it follows from $|fv|\ge |z_1|$ that $fv\in A^*z_1$. 
Accordingly both the words  $z_1v$ and $z_1$ are suffixes of $f$: we obtain $z_1v \in A^*z_1$.
Since $z_1$ is an overlapping-free word, by definition the cond. $|v|=|v|+|v'|\le |z_0|\le |z_1|-1$ implies $v=\varepsilon=v'$.
\begin{figure}[H]
\begin{center}
\includegraphics[width=9.5cm,height=4.5cm]{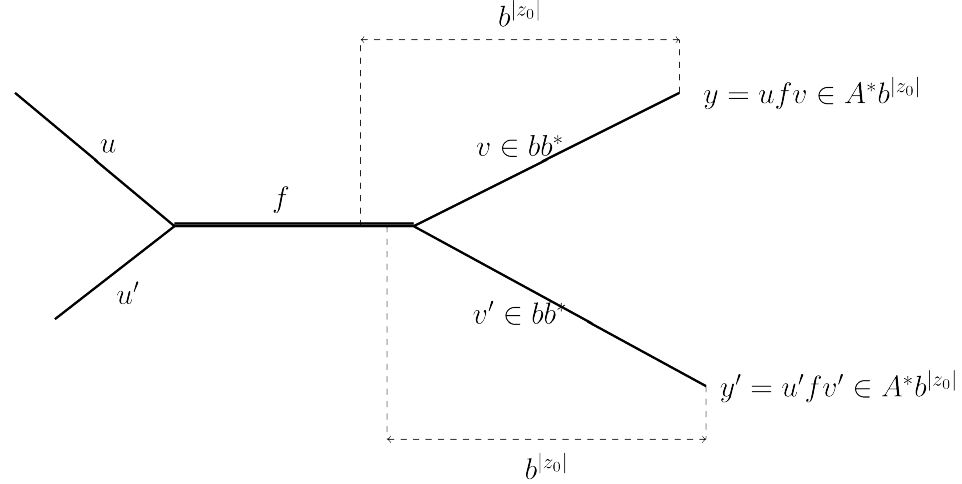}
\end{center}
\caption[]{Proof of Prop. \ref{Com-Fact--Independent}. Case (c1): we have $(y,y')\in \uline{{\cal F}_k}$, with  $v\ne\varepsilon$ and $y'=v'\ne\varepsilon$.} 
\label{FigureFactorComp-b1}
\end{figure}
\smallskip
\noindent\hspace*{8mm}(c2) Now, we compare $u$, $u'$ with  $\varepsilon$. In order to do so we apply to the prefixes of $y$ and $y'$ arguments similar to the preceding ones.
Firstly, by contradiction, we assume $u\ne\varepsilon$ and $u'\ne \varepsilon$. It follows from  $2\le |u|+|u'|\le k\le |z_0|$ and $u,u'\in {\rm P}(Y_1)$
that $u,u'\in {\rm P}(a^{|z_0|})\setminus\varepsilon$.
Accordingly  $i,j\in [1,|z_0|]$ exist st.  $u=a^i$, $u'=a^j$, thus $af\in {\rm F}(y)\cap {\rm F}(y')$: a  contradiction with $f$ being  of maximal length in ${\rm F}(y)\cap {\rm F}(y')$.
Consequently, at least one of the two words $u$, $u'$ is the empty word: wlog. we assume $u'=\varepsilon$, thus $y'=fv'=f$.
As observed in the case (c1), according to the maximality of $|f|$ we have  $|f|\ge |z_1|$.
Since both the words $f$ and $z_1$ are prefixes of $y'$, this implies $f\in z_1A^*$, thus $uf\in uz_1A^*$.
In addition, by construction both the words $z_1$ and $uf$ are prefixes of $y$: this implies $uf\in uz_1A^*\cap z_1A^*$ (see Fig. \ref{FigureFactorComp-b2}) that is, $uz_1\in {\rm P}(F)$ and  $z_1\in {\rm P}(F)$, thus $uz_1\in z_1A^*$.
Since $z_1$ is  overlapping-free, we obtain $u=\varepsilon=u'=v=v'$, thus $y=y'$: a contradiction with $(y,y')\in \uline{{\cal F}_k}$.  
As a consequence, we have $\uline{{\rm \cal F}_k}(Y_1)\cap Y_1=\emptyset$: this completes the proof.
\end{proof}
\begin{figure}[H]
\begin{center}
\includegraphics[width=7cm,height=3.cm]{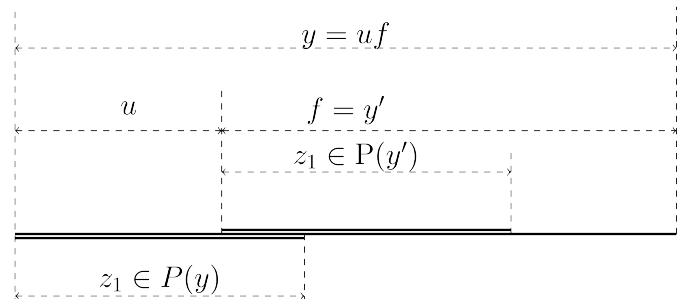}
\end{center}
\caption[]{Proof of Prop. \ref{Com-Fact--Independent}: Case (c2) $(y,y')\in \uline{{\cal F}_k}$, with $v=v'=\varepsilon$, and $u'=\varepsilon$.} 
\label{FigureFactorComp-b2}
\end{figure}
\noindent According to Prop. \ref{Com-Fact--Independent}, by merely substituting ${\rm\cal F}_k$ to ${\rm\cal P}_k$ in the proofs of  the statements  \ref{classic11-Pi-k}--\ref{c3_Pi},  we obtain the following result:
\begin{proposit}
\label{classic-22-Phi-k}
Given a  regular code $X$, each of the  following properties holds:

{\rm (i)} The three following conds. are equivalent:

\hspace*{4mm}-- $X$ is maximal in the family of  $\uline{{\rm \cal F}_k}$-independent codes.

\hspace*{4mm}-- $X$ is  complete.

\hspace*{4mm}-- $\mu(X)=1$.

{\rm (ii)} One  can decide whether $X$ satisfies Cond. \ref{3} wrt.  ${\rm \cal F}_k$.

{\rm (iii)} If $X$ is $\uline{{\rm \cal F}_k}$-independent, it can be embedded into a maximal $\uline{{\cal F}_k}$-independent code.
\end{proposit}
%
\section{Error detection wrt. topologies associated with (anti-)automorphisms}
\label{antiautomorphism}
Let $\theta$ be  (anti-)automorphism of $A^*$. As justified in the introduction,  we focus on the relation $\tau_{d_\theta,1}=\widehat\theta=\theta\cup id_{A^*}$.
Recall that  we have $\uline{\left(\widehat\theta\right)}=\widehat\theta\setminus id_{A^*}=\uline\theta$. 
%
%
\begin{example}
\label{theta1}
{\rm 
 Let  $A=\{a,b\}$ and $\theta$ be the automorphism defined by $\theta(a)=b$, $\theta(b)=a$.
It follows from  $\theta^{-1}(a)=b=\theta(a)$ 
that $\theta^{-1}=\theta=\uline\theta$.
Consider the  regular prefix code $X=\{a^nb: n\ge 0\}$.

\noindent\hspace*{2mm} -- $X$  satisfies Cond. \ref{1} wrt.  $\tau_{d_\theta,1}=\widehat\theta$.
 Indeed, we have $\uline{\left(\widehat\theta\right)}(X)=\uline\theta(X)=\theta(X)=\{b^na: n\ge 0\}$, thus $\uline\theta(X)\cap X=\emptyset$.

\noindent\hspace*{2mm} -- 
We have $\left(\widehat\theta\right)^{-1}= \left(\theta\cup id_{A^*}\right)^{-1}=\theta^{-1}\cup id_{A^*}^{-1}=\theta\cup id_{A^*}=\widehat\theta$.
Furthermore we obtain  $\widehat\theta\cdot\left(\widehat\theta\right)^{-1}=(\theta\cup id_{A^*})^2=\theta^2\cup\theta\cup id_{A^*}=\theta\cup id_{A^*}=\widehat\theta$.
As indicated above, $X$ satisfies Cond. \ref{1} wrt. $\widehat\theta=\widehat\theta\cdot\left(\widehat\theta\right)^{-1}$:
according to Prop. \ref{equiv-c2}, $X$ also satisfies Cond. \ref{2} wrt.  $\widehat\theta$.

\noindent\hspace*{2mm}  -- According to Theorem \ref{classic}, it follows from $\mu(X)=\frac{1}{2}\sum_{n\ge 0}\left(\frac{1}{2}\right)^n=1$ that $X$ is a maximal prefix code,
therefore $X$ is maximal in the family of $\uline{\left(\widehat\theta\right)}$-independent codes (Cond. \ref{3}).  

\noindent\hspace*{2mm}  -- Finally, it follows from $ba\in \theta(X)\setminus X$ that $X\subsetneq\widehat\theta(X)$.
Consequently, since $X$ is a maximal code,  $\widehat\theta(X)$ cannot be  a code that is, $X$ cannot satisfy Cond. \ref{4} (we verify that $a,b,ab\in \widehat\theta(X)$).
} 
\end{example}
\begin{example}
\label{theta2}
{\rm
Over the alphabet  $A=\{a,b\}$,  take for $\theta$ the anti-automorphism defined by $\theta(a)=b$, and $\theta(b)=a$ and, once more, consider  the code $X=\{a^nb: n\ge 0\}$.

\noindent\hspace*{2mm}  -- If follows from $\theta(X)\cap X=\{ab\}$, that the equation $x=\theta\left(x\right)$ is equivalent to $x=ab$.
This implies  $\uline{\theta}(X)=\{ab^n: n\ne 1\}$, thus $\uline{\theta}(X)\cap X=\emptyset$, whence $X$  satisfies Cond. \ref{1}.

\noindent\hspace*{2mm} -- As in Example \ref{theta1}, we have  $\theta^{-1}=\theta$. 
Once more  this implies $\widehat\theta\cdot\left(\widehat\theta\right)^{-1}=\widehat\theta$, thus $X$ satisfies Cond. \ref{2} wrt.  $\widehat\theta$.

\noindent\hspace*{2mm}  -- Similarly, we have $\mu(X)=1$, whence $X$ satisfies Cond. \ref{3}.  

\noindent\hspace*{2mm}  -- Lastly, it follows from $a^2b\in \theta(X)\setminus X$ that  $X\subsetneq \widehat\theta(X)$: since $X$ is a maximal code, $\widehat\theta(X)$ cannot be a code that is, $X$ cannot satisfy Cond. \ref{4}.
}
\end{example}
\begin{example}
\label{theta4}
{\rm
 Over the alphabet $\{A,C,G,T\}$, let  $\theta$ denotes  the Watson-Crick anti-automorphism (see eg. \cite{KKK14,K83}),
which is defined by  $\theta(A)=T$, $\theta(T)=A, \theta(C)=G$, and $\theta(G)=C$. We have $\theta^{-1}=\theta=\uline\theta=\uline{\left(\widehat\theta\right)}$.
Consider the prefix code $X=\{A,C,GA,G^2,GT,GCA,GC^2,GCG,GCT\}$.

\noindent -- It follows from $\theta(X)=\{T,G,TC,C^2,AC,TGC,G^2C, CGC,AGC\}$ that $\uline{\left(\widehat\theta\right)}(X)\cap X=\emptyset$, whence satisfies Cond. \ref{1}.

\noindent --  As in the examples \ref{theta1}, \ref{theta2}, it follows from  $\theta^{-1}=\theta$ that  $\widehat\theta \cdot\left(\widehat\theta\right)^{-1}=\widehat\theta$,
therefore, $X$ satisfies Cond. \ref{2} wrt.  $\tau_{d_\theta,1}=\widehat\theta$.

\noindent -- We have $\mu(X)=2/4+ 3/4^2+4/4^3=3/4<1$, hence $X$ cannot satisfy Cond. \ref{3}.

\noindent -- At last,  it follows from $G,G^2\in\widehat\theta (X)=\theta(X)\cup X$ that  Cond. \ref{4} is not  satisfied.
}
\end{example}
\begin{example}
\label{theta5}
{\rm
In each of the preceding examples, since the mapping $\theta$ satisfies $\theta^{-1}=\theta$, the quasi-metric  $d_\theta$ is actually a  metric.
Of course, $d_\theta$ could be only a quasi-metric.
For instance over $A=\{a,b,c\}$, taking for $\theta$  the automorphism generated by the cycle $(a,b,c)$, we obtain $d_\theta(a,b)=1$ and $d_\theta(b,a)=2$ 
(we have $b=\theta(a)$ and $a\ne\theta(b)$).
}
\end{example}
\subsection{Questions involving  regular sets}
We start with the following result:
\begin{proposit}
\label{theta-X-regular}
With the preceding  notation, each of the following properties holds:

{\rm (i)} If $\theta$  is an  automorphism, then the relations $\tau_{d_{\theta,1}}=\widehat\theta$ and $\uline{\tau_{d_{\theta,1}}}=\uline\theta$  are 
regular.

{\rm (ii)} If $\theta$ is an anti-automorphism, then $\widehat\theta$ is a non-regular relation.

{\rm (iii)} In any case,  $\widehat\theta$ is regularity-preserving.
\end{proposit}
\begin{proof}
-- In the case where  $\theta$  is an  automorphism of $A^*$, it is a regular relation:
indeed,  $\theta$ is the behavior of the one-state automaton with transitions $q_0\xrightarrow {\left(a,\theta(a)\right)}q_0$ for all $a\in A$.
Set $B=\{a\in A: \theta(a)\ne a\}$.
Starting with the preceding automaton,  we obtain an automaton with behavior $\widehat\theta$
  by merely adding  the transitions $\left(q_0,\left(a,a\right),q_0\right)$, for all $a\in B$ (see  Fig. \ref{FigureAutomaton-hat-theta}): in other words, $\widehat\theta$ is a regular relation.
\begin{figure}[H]
\begin{center}
\includegraphics[width=2.4cm,height=1.5cm]{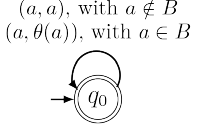}
\end{center}
\caption[]{\small An automaton with behavior $\tau_{d_{\theta,1}}=\widehat\theta$ in the case where  $\theta$ is an automorphism.}
\label{FigureAutomaton-hat-theta}
\end{figure}
\noindent The relation $\uline{\left(\widehat\theta\right)}=\uline\theta$, for its part, is the set of all the pairs $(w,w')$ with $w\ne w'$ and st.  $w'=\theta(w)$. By construction we have $|w|=|w'|$. Let $u=w\wedge w'$. Necessarily, we have $u\in (A\setminus B)^*$, moreover $a, b\in B$ and $s,s'\in A^*$ exist such that $w=uas$, $w'=ubs'$. 
Since $\theta$ is a free monoid automorphism, necessarily we have $b=\theta(a)$ and $s'=\theta(s)$.
Accordingly, $\uline\theta$ is the  behavior of the two-state automaton provided in Fig. \ref{FigureAutomaton-Under-theta}, whence it is a regular relation.
\begin{figure}[H]
\begin{center}
\includegraphics[width=5.5cm,height=1.3cm]{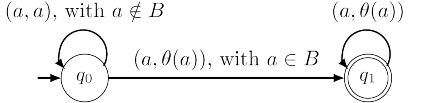}
\end{center}
\caption[]{\small An automaton with behavior $\uline\theta$, in the case where $\theta$ is an automorphism.}
\label{FigureAutomaton-Under-theta}
\end{figure}
\noindent\hspace*{2mm} 
-- Let $\theta$ be an anti-automorphism of $A^*$. Classically, an automorphism of $A^*$, say $h$, exists
st.   $\theta=t\cdot h$,
where  $t:w\rightarrow w^R$, is the so-called {\it transposition} onto $A^*$. 
In addition  $t$ cannot be  a regular relation  (see eg. \cite [Example IV.1.10]{S03}).
By contradiction, assume $\theta $ being regular. Since we have $t= \theta\cdot h^{-1}$, and since $h^{-1}$ itself is a monoid automorphism,
 according to Prop. \ref{property-reg} the transposition $t$ should be a  regular relation: a contradiction with the above.

\smallbreak
\noindent\hspace*{2mm} 
-- For proving the property (iii), we consider a regular set $X\subseteq A^*$. In the case where $\theta$ is an automorphism, the relation $\widehat\theta=\theta\cup id_{A^*}$ is regular: according to Prop. \ref{property-reg}, $\widehat\theta(X)$ is regular.
In the case where $\theta$ is an anti-automorphism,  with the preceding notation 
although the transposition $t$ is not a regular relation, the set $t(X)$ is classically known to be regular  (see eg. \cite[Proposition I.1.1] {S03}).
According to Prop. \ref{property-reg}, the set  $\widehat\theta(X)=h\left(t(X)\right)$ itself is regular. 
\end{proof}
\noindent
\subsection{Decidability results forConds. \ref{1}--\ref{4}}
Regarding Conds. \ref{1}, \ref{2}, (anti-)automorphisms satisfy the following noticeable property:
\begin{lem}
\label{1-2-equiv}
A regu
 $X\subseteq A^*$ satisfies Cond. \ref{1} wrt.  $\widehat\theta$,  iff.  it satisfies Cond. \ref{2}. 
\end{lem}
\begin{proof}
Firstly, assume  that $X$ is  $\uline\theta$-independent, and let $x,y\in X$ st.  $\uline{\tau_{d_\theta,1}}\left(x\right)\cap \uline{\tau_{d_\theta,1}}\left(y\right)=\uline\theta\left(x\right)\cap\uline\theta\left(y\right)\ne\emptyset$: 
necessarily we have  $\theta\left(x\right)\cap\theta\left(y\right)\ne\emptyset$. Since $\theta$ is a one-to-one mapping, this implies $x=y$
therefore, by definition
$X$ satisfies Cond. \ref{2}.
For proving the converse, we argue by contrapositive.  Assuming that Cond. \ref{1} does not hold that is, 
$X\cap\uline{\left(\widehat\theta\right)}(X)=X\cap\uline\theta(X)\ne\emptyset$, 
a pair of words $x,y\in X$ exist st.  $y=\theta\left(x\right)$, with $x\ne y$.
It follows from $\widehat\theta\left(x\right)=\{x\}\cup\{\theta\left(x\right)\}=\{x,y\}$ and  $\widehat\theta\left(y\right)=\{y\}\cup\{\theta\left(y\right)\}$ that $\widehat\theta\left(x\right)\cap \widehat\theta\left(y\right)\ne\emptyset$, whence Cond. \ref{2} cannot hold.
\end{proof}
\noindent
As a consequence of  Props.  \ref{theta-X-regular}, \ref{1-2-equiv}, we obtain the following result:
\begin{proposit}
\label{decide-theta}
Given $X\subseteq A^*$, wrt. $\widehat\theta$  each of the following properties holds:

{\rm (i)} In any case, $X$ satisfies both Conds. \ref{1},\ref{2}.

{\rm (ii)} If $X\subseteq A^*$  is  a regular code, it can be decided whether it satisfies Cond. \ref{4}.
\end{proposit}
\begin{proof}
-- For proving that $X$ satisfies Cond. \ref{2}, we prove that it satisfies the cond. (ii) of  Lemma \ref{equiv-c2}.
Actually, it follows from $\theta\subseteq\widehat\theta$ that $\theta\cdot\theta^{-1}\subseteq \widehat\theta\cdot\widehat\theta^{-1}$.
Since $\theta$ is a one-to-one mapping, we have $\theta\cdot\theta^{-1}=id_{A^*}$: this implies $\widehat\theta\cdot\widehat\theta^{-1}=id_{A^*}$. 
As a consequence, for every $x\in X$ we have $\widehat\theta\cdot\left(\widehat\theta\right)^{-1}\left(x\right)\cap X=\{x\}\cap X=\{x\}$, thus $\left(\uline{\widehat\theta\cdot\left(\widehat\theta\right)^{-1}} \right)\left(x\right)=\emptyset$.
According to  Lemma \ref{equiv-c2} $X$ satisfies Cond. \ref{1} wrt.  $\widehat\theta$ that is, 
according to Lemma \ref{1-2-equiv}, $X$ also satisfies Cond. \ref{1}.

\smallskip\noindent\hspace*{2mm} 
-- According to Prop. \ref{theta-X-regular}, in any case the set $\theta(X)$  is regular therefore, according to Prop. \ref{decidable-bases},
 one can  decide whether $\widehat\theta(X)$ is a code that is, whether $X$ satisfies Cond. \ref{4}.
\end{proof}
\noindent
Before to study the behavior of  Cond. \ref{3}, we note that the following property  holds: 
\begin{fait}
\label{claim4}
If $\theta$ is an anti-automorphism,  for all $w,w'\in A^*$, $w'\in {\rm F}(w)$ implies $\theta(w')\in {\rm F}(w)$.
\end{fait}
\begin{proof}
 Let $u,v\in A^*$ st. $w=uw'v$. By definition, we have $\theta(w)=\theta(v)\theta(w')\theta(u)$.
\end{proof}
\noindent
As usual in the paper, we start by examining completeness from the point of view of embedding:
\begin{proposit}
\label{Embed-theta}
Every regular $\uline\theta$-independent code can be embedded into some complete one.
\end{proposit}
\begin{proof}
The property trivially holds if $X$ is complete. Assume $X$ non-complete.
According to Theorem \ref{EhRz}, the result holds if $\theta$ is an automorphism:
indeed the action of such a transformation merely consists in rewriting words by applying some permutation of $A$.
In the sequel we assume that $\theta$ is  an anti-automorphism.
Classically, some positive integer $n$, the so-called {\it order} of the permutation $\theta$,  exists such $\theta^{n}=id_{A^*}$. 
As in  the proofs of Prop.  \ref{Comp-Pref-Independent} and \ref{Com-Fact--Independent}, in view of Theorem \ref{EhRz}, we start by constructing a convenient word in $A^*\setminus {\rm F}(X^*)$.
Let $z_0\notin {\rm F}(X^*)$, let $a$ be its initial character, and let $b$ be a character different of $a$.
We assume wlog. $|z_0|\ge 2$ and $z_0\notin aa^*$, for every $a\in A$ (otherwise,  substitute $z_0b$ to $z_0$).
Since $\theta$ is a free monoid anti-automorphism, it is length-preserving that is, for every  $i\ge 0$, the equation $\left|\theta^i(z_0)\right|=\left|\theta(z_0)\right|$ holds.
Consequently, we have  $\left|z_0\theta(z_0)\cdots\theta^{n-1}(z_0)\right|=n|z_0|$ therefore,
according to Lemma \ref {overlapping-free constr}  $z_2=z_0\theta(z_0)\cdots\theta^{n-1}(z_0)ab^{n|z_0|}$ is an overlapping-free word. 
In addition, it follows from $z_0\in A^*\setminus {\rm F}(X^*)$ that 
 $z_2\in  A^*\setminus {\rm F}(X^*)$.
Set $U_2=A^*\setminus \left(X^*\cup A^*z_2A^*\right)$, $Y_2=(z_2U_2)^*z_2$, and $Z_2=X\cup Y_2$.
According to  Theorem \ref{EhRz}, the set $Z_2=X\cup Y_2$ is a complete regular code.
  Since we assume $X$ being $\widehat\theta$-independent, 
we have $\uline\theta(Z_2)\cap Z_2=\left(\uline\theta(X)\cap Y_2\right)\cup \left(X\cap \uline\theta(Y_2)\right)\cup \left(\uline\theta(Y_2)\cap Y_2\right)$.
In order to prove that $Z_2$ is $\uline\theta$-independent, we argue by contradiction.
Actually assuming that $\uline\theta(Z_2)\cap Z_2\ne \emptyset$, exactly one of the three following conds. holds:

\smallskip
\noindent\hspace*{4mm}
(a)  {\it Cond. $\uline\theta(X)\cap Y_2\ne\emptyset$.} 
With this condition, $x\in X$ exists st.  $\theta(x)\in Y_2$. 
Since by construction $z_2$ is a prefix of any word in $Y_2$, we have $z_2\in {\rm F}\left(\theta\left(x\right)\right)$.
In addition, it follows from  $\theta(z_0)\in {\rm F}(z_2)$ that  $\theta(z_0)\in {\rm F}\left(\theta\left(x\right)\right)$.
According to Claim \ref{claim4},  we obtain $\theta^{n}(z_0)\in {\rm F}\left(\theta^{n}\left(x\right)\right)$, thus $z_0\in {\rm F}(x)$,
a contradiction with $z_0\notin {\rm F}(X^*)$.

\smallskip
\noindent\hspace*{4mm}
(b) {\it Cond. $X\cap \uline\theta(Y_2)\ne\emptyset$.} Some  pair of words $x\in X$, $y\in Y_2$ exist st.  
$x=\theta(y)$.
By construction we have $\theta^{n-1}(z_0)\in {\rm F}(z_2)\subseteq {\rm F}(y)$.
According to Claim \ref{claim4} this implies $\theta\left(\theta^{n-1}(z_0)\right)\in {\rm F}\left(\theta\left(y\right)\right)$, thus $z_0\in {\rm F}\left(x\right)$: once more this contradicts
$z_0\notin {\rm F}(X^*)$.

\smallskip
\noindent\hspace*{4mm}
(c)   {\it Cond. $\uline\theta(Y_2)\cap Y_2\ne\emptyset$.} With this cond. there are  different words $y, y'\in Y_2$ st.  $y'=\theta\left(y\right)$.
On the one hand, since $\theta$ is an anti-automorphism,  $b^{n|z_0|}\in {\rm S}(y')$ implies $\theta\left(b^{n|z_0|}\right)\in {\rm P}(y')$.
On the other hand, it follows from    $Y_2\subseteq z_0A^*$ and $y'\in Y_2$ that $z_0\in{\rm P}(y')$.
More precisely,  $|z_0|\le  |b^{n|z_0|}|$ implies  $z_0\in{\rm P}(b^{n|z_0|})$, thus 
$z_0=\left(\theta\left(b\right)\right)^{|z_0|}$. From the fact that we have  $|z_0|\ge 2$ and $\theta(b)\in A$,
this is incompatible with the construction of $z_0$.

\smallskip\noindent
In each case we obtain a contradiction, whence $Z_2$ is $\uline{\theta}$-independent: this completes the proof.
\end{proof}
\noindent
As a consequence, we obtain the following result-the proof is merely done by translating in term of (anti-)automorphism the one of Prop.  \ref{classic-22-Phi-k} (recall that we have  $\uline{\left(\widehat\theta\right)}= \uline\theta$):
\begin{proposit}
\label{classic-33-theta-1}
Given a  regular code $X$, each of the  following properties holds:

{\rm (i)} The three following conds. are equivalent:

\hspace*{4mm}-- $X$ is maximal as a $\uline\theta$-independent code.

\hspace*{4mm}-- $X$ is  complete.

\hspace*{4mm}-- $\mu(X)=1$ holds.

{\rm (ii)} One  can decide whether $X$ satisfies Cond. \ref{3} wrt.  ${\widehat\theta}$. 

{\rm (iii)} If $X$ is $\uline\theta$-independent, then  it can be embedded into a maximal $\uline\theta$-independent code.
\end{proposit}
\noindent
We close the study with the following statement: it synthesizes the decidability results obtained in the whole paper:
\begin{theorem}
\label{decidable}
With the preceding notation,  each of the following properties holds:

{\rm (i)} 
For every positive integer $k$, it can be decided whether $X$ satisfies any of Conds. \ref{1}--\ref{4} wrt.  ${\cal P}_k$,  ${\cal S}_k$, or ${\cal F}_k$.

{\rm (ii)} 
For every  (anti-)automorphism $\theta$ of $A^*$, wrt.  $\tau_{d_\theta,1}=\widehat\theta$, the code $X$ satisfies Conds. \ref{1} and \ref{2}.
In addition one can decide whether it  satisfies any of Conds. \ref{3}, \ref{4}.
\end{theorem}
\section{Concluding remark}
\label{conclusion}

From the point of view of decidability, we have now fully studied  the behaviors of Conds. \ref{1}--\ref{4} wrt.  ${\cal P}_k$, ${\cal F}_k$, and $\widehat\theta$.
We also note  that, in \cite{KM15} the  authors were interested in the question of embedding a non maximal $\tau$-independent set $L\subseteq A^*$ into some maximal one.
From this point of view, the results of Corollary  \ref{Comp-ind-max} and Props.  \ref{classic-22-Phi-k}, \ref{classic-33-theta-1} attest that,  in the frameworks of ${\cal P}_k$, ${\cal P}_k$, and $\widehat\theta$, in any case such an embedding can be successfully done
for variable-length codes.

\smallskip\noindent
Regarding further research, several ways may be involved:

\noindent\hspace*{4mm}
--  At first, for $k\ge 2$ the question whether the relation $\uline{{\cal F}_k}$ is regular or not remains open.

\noindent\hspace*{4mm}
-- The  so-called {\it subsequence metric} associates, with  each pair of words $(w,w')$, the integer $\delta(w,w')=|w|+|w'|-2|lcs(w,w')|$,
where $lcs(w,w')$ stands for some maximum length subsequence  common  to $w$ and $w'$.
With such a definition that metric appears as a direct extension of the factor one.
Equivalently, $\delta(w,w')$ is the minimum number of  one character insertions and  deletions that have to be applied   for computing $w'$ by starting from $w$. 
From this point of view, the frameworks of the Hamming and Levenshtein metrics are also involved.

We observe that, wrt.  $\tau_{\delta,k}$, results very similar to  Prop.    \ref{Phi-regular}, \ref{classic-22-Phi-k} have been  established  in \cite{N21},
however,  the question whether or not Conds. \ref{1}, \ref{2} are decidable remains open. 

\noindent\hspace*{4mm}
 -- More generally, it appears natural to study, from the point of view of decidability, the behavior of each of the conds. \ref{1}--\ref{4}  in the framework of other topologies in the free monoid. 
Without being exhaustive, we mention
the so-called additivity preserving  quasi-metrics  \cite{CSY01},  or metrics based on absent words \cite{CMR20}. 
 
\noindent\hspace*{4mm}
 -- From another point of view, a quasi-metric being fixed over $A^*$, presenting families of codes satisfying all the best Conds. \ref{1}--\ref{4} would be desirable.
\bibliographystyle{plain}
\bibliography{TOPOLOGY}

\begin{thebibliography}{10}

\bibitem{BPR10}
J.~Berstel, D.~Perrin, and C.~Reutenauer.
\newblock {\em Codes and Automata}.
\newblock Cambridge University Press, 2010.

\bibitem{BWZ90}
V.~Bruy\`ere, Limin Wang, and Liang Zhang.
\newblock On completion of codes with finite deciphering delay.
\newblock {\em European J. Comb.}, 11:513--521, 1990.

\bibitem{CSY01}
C.S. Calude, A.~Salomaa, and S.~Yu.
\newblock Additive distances and quasi-distances between words.
\newblock {\em J. of Universal Comput. Sci.}, 8(2):141–122, 2001.

\bibitem{CMR20}
G.~Castiglione, S.~Mantaci, and A.~Restivo.
\newblock Some investigations on similarity measures based on absent words.
\newblock {\em Fundam.. Informaticae}, 171:97--112, 2020.

\bibitem{CP02}
C.~Choffrut and G.~Pighizzini.
\newblock Distances between languages and reflexivity of relations.
\newblock {\em Theoret. Comp. Sci.}, 286:117--138, 2002.

\bibitem{C81}
P.M. Cohn.
\newblock {\em Universal Algebra (Mathematics and Its Applications, 6)}.
\newblock Springer, 1981.

\bibitem{ER85}
A.~Ehrenfeucht and S.~Rozenberg.
\newblock Each regular code is included in a regular maximal one.
\newblock {\em RAIRO - Theoret. Informatics and Appl.}, 20:89--96, 1986.

\bibitem{E74}
S.~Eilenberg.
\newblock {\em Automata, Languages and Machines}.
\newblock Academic Press, 1974.
\newblock eBook {I}SBN: 9780080873749.

\bibitem{EM65}
C.C. Elgot and J.~Meizei.
\newblock On relations defined by generalized finite automata.
\newblock {\em IBM J. Res. Develop.}, 9:47--68, 1965.

\bibitem{GHK22}
H.~Gruber, M.~Holzer, and M.~Kutrib.
\newblock Descriptional complexity of regular languages.
\newblock In J-\'E. Pin, editor, {\em Handbook of Automata Theory}, volume~I,
  chapter~12, pages 411--458. EMS Press, Berlin, 2021.

\bibitem{H50}
R.W. Hamming.
\newblock Error detecting and error correcting codes.
\newblock {\em The Bell Technical Journal}, 29:147--160, 1950.

\bibitem{HU79}
J.~E. Hopcroft and D.~Ullman.
\newblock {\em Introduction to automata theory, languages and Computation}.
\newblock Addison-Wesley publishing company Reading, MA, Menlo Park, CA,
  London, Amsterdam, Don Mills ONT, Sydney, 1979.

\bibitem{JKK01}
H.~J\"urgensen, M.~Katsura, and S.~Konstantinidis.
\newblock Maximal solid codes.
\newblock {\em J. of Automata, Languages and Combinatorics}, 6:25--50, 2001.

\bibitem{JK97}
H.~J\"urgensen and S.~Konstantinidis.
\newblock Codes.
\newblock In {\em Handbook of Formal Languages}, volume~1, chapter~8, pages
  511--607. Springer Verlag, Berlin, 1997.
\newblock {I}SBN 78-3-642-59136-5.

\bibitem{KKK14}
L.~Kari, S.~Konstantinidis, and S.~Kopecki.
\newblock On the maximality of languages with combined types of code
  properties.
\newblock {\em Theoret. Comp Sci.}, 550:79--89, 2014.

\bibitem{KM15}
S.~Konstantinidis and M.~Mastnak.
\newblock Embedding rationally independent languages into maximal ones.
\newblock {\em J. of Aut., Lang. and Comb.}, 21:311--338, 2016.

\bibitem{K83}
J.~Kruskal.
\newblock An overview of sequence comparison: Time warps, string edits, and
  macromolecules: The theory and practice of sequence comparison.
\newblock {\em SIAM J. Comput.}, 25:201--234, 1983.

\bibitem{L00}
N.H. Lam.
\newblock Finite maximal infix codes.
\newblock {\em Semigr. Forum}, 61:346--356, 2000.

\bibitem{L01}
N.H. Lam.
\newblock Finite maximal solid codes.
\newblock {\em Theot. Comput. Sci.}, 262:333--347, 2001.

\bibitem{L03}
N.H. Lam.
\newblock Completing comma-free codes.
\newblock {\em Theot. Comput. Sci.}, 301:399--415, 2003.

\bibitem{L65}
V.I. Levenshtein.
\newblock Binary codes capable of correcting deletions, insertion and
  reversals.
\newblock {\em Soviet Physics Doklady - Cybernetics and Control Theory},
  10(8):107--110, 1966.
\newblock {\rm transl. from:} Doklady Academii Nauk SSSR Vol. 163, No 4, pp.
  845-848, August, 1965.

\bibitem{Lo1983}
M.~Lothaire.
\newblock {\em Combinatorics on Words}.
\newblock Addison-Wesley Publishing Company (2nd edition Cambridge University
  Press 1997), 1983.

\bibitem{MY60}
R.~McNaughton and H.~Yamada.
\newblock Regular expressions and state graphs for automata.
\newblock {\em IRE Trans. Electronic Computers}, 9:39--47, 1960.

\bibitem{N06}
J.~N\'eraud.
\newblock On the completion of codes in submonoids with finite rank.
\newblock {\em Fundam. Informaticae}, 74:549--562, 2006.
\newblock {I}SSN 0169-2968.

\bibitem{N08}
J.~N\'eraud.
\newblock Completing circular codes in regular submonoids.
\newblock {\em Theoret. Comput. Sci.}, 391:90--98, 2008.

\bibitem{N21}
J.~N\'eraud.
\newblock Variable-length codes independent or closed with respect to edit
  relations.
\newblock {\em Inf. Comput.}, 288, 2022.

\bibitem{N22}
J.~N\'eraud.
\newblock When variable-length codes meet the field of error detection.
\newblock In D.~Poulakis, G.~Rahonis, and P.~Tzounakis, editors, {\em 9th
  International Conference on Algebraic Informatics: CAI 2022}, volume 13706,
  pages 203--222. Lect. Notes in Comp. Sci., 2022.

\bibitem{NS20}
J.~N\'eraud and C.~Selmi.
\newblock Embedding a $\theta$-invariant code into a complete one.
\newblock {\em Theoret. Comput. Sci.}, 806:28--41, 2020.

\bibitem{Ng16}
T.~Ng.
\newblock Prefix distance between regular languages.
\newblock In {\em Implementation and Applications of Automata}, volume 9705,
  pages 224--235. Lect. Notes in Comp. Sci., 2016.

\bibitem{Ni68}
M.~Nivat.
\newblock Transductions des langages de chomsky.
\newblock {\em Ann. Inst. Fourier (Grenoble)}, 18:339–455, 1968.

\bibitem{R77}
A.~Restivo.
\newblock On codes having no finite completion.
\newblock {\em Discr. Math.}, 17:309--316, 1977.

\bibitem{S03}
J.~Sakarovitch.
\newblock {\em Elements of Automata Theory}.
\newblock Cambridge University Press, 2009.

\bibitem{S22}
J.~Sakarovitch.
\newblock Automata and rationnal expressions.
\newblock In J-\'E. Pin, editor, {\em Handbook of Automata Theory}, volume~I,
  chapter~2, pages 39--78. EMS Press, Berlin, 2021.

\bibitem{SP53}
A.~Sardinas and G.~W. Patterson.
\newblock A necessary and sufficient condition for the unique decomposition of
  coded messages.
\newblock {\em IRE Internat Con. Rec.}, 8:104--108, 1953.

\bibitem{VVH05}
Do~Long Van, Kieu Van~Hung, and Phan~Trung Huy.
\newblock Codes and length-increasing transitive binary relations.
\newblock In Dang Van~Hung and Martin Wirsing, editors, {\em Theoretical
  Aspects of Computing -- ICTAC 2005}, pages 29--48, Berlin, Heidelberg, 2005.
  Springer Berlin Heidelberg.

\bibitem{W31}
W.~A. Wilson.
\newblock On {Q}uasi-{M}etric {S}paces.
\newblock {\em American J. of Math.}, 53:675--684, 1931.

\bibitem{ZS95}
Liang Zhang and Zhonghui Shen.
\newblock Completion of recognizable bifix codes.
\newblock {\em Theoret. Comput. Sci.}, 145:345--355, 1995.

\end{thebibliography}

\section*{Appendix}
\label{deeper-decidability}
In view   of Theorem \ref{decidable}, in what follows we provide some outline of basic results in order to implement corresponding algorithms:

\medbreak\noindent
(i) {\it Regular operations,  boolean operations}

The proof of Prop. \ref{properties-reg-A} lays upon the fact that every regular (resp., boolean) operation among regular sets
can be translated in term of corresponding ones among finite automata (see eg. \cite[Chapter 3]{HU79}). 
Operations such as  quotient, direct,  or inverse image under monoid endomorphism  are also involved. 

\medbreak\noindent
(ii)  {\it Deciding whether a regular set is empty}

In order to do so, starting with a finite automaton with behavior $X$, we will decide whether  or not  a successful path exists by applying some classical  graph mining algorithm. 
More precisely, according to \cite[Proposition 3.7]{HU79}, for a $n$-state automaton, we have $X=\emptyset$ iff.  no path with length less than $n$ can be successful.

\medbreak\noindent
(iii) {\it Regular expressions and automata for regular subsets of $A^*$}

A description of a regular set $X\subseteq A^*$ by using only the operations union, product and Kleene star is called a {\it regular expression} \cite[Sect. 2.5]{HU79}.
Several classical methods can be applied in order to switch between the representation of $X$ by an automaton or by a corresponding regular expression\cite{GHK22,S22}.
Furthermore, every regular subset of $A^*$  can be described by a so-called {\it unambiguous} regular expression, which is exclusively built from unambiguous regular operations (see eg. \cite[Corollary VII.8.3]{E74}).
In order to compute such an unambiguous regular expression, starting with a finite automaton with behavior $X$, several methods can be applied, 
the best-known certainly  being 
McNaughton and Yamada algorithm (see \cite{MY60} or \cite[Proposition 4.1.8]{BPR10}). 
Regarding Bernoulli measure,  unambiguous regular expressions allow to compute $\mu(X)$ by recursively applying the three following formulas: 
$\mu(R_1+R_2)=\mu(R_1)+\mu(R_2)$, $\mu(R_1R_2)=\mu(R_1)\mu(R_2)$, and $\mu(R^*)=(1-\mu(R))^{-1}$.

\medbreak\noindent
(iv) {\it Image of a regular set under a regular word binary relation}.

The proof of  Prop. \ref{property-rec1},  lays upon a result from  \cite{Ni68}, concerning a peculiar decomposition of regular relations. 
More precisely, given a regular relation $\tau\subseteq A^*\times A^*$,  there are a finite alphabet $Z$, a regular set $K\subseteq  Z^*$, and two monoid homomorphisms $\phi,\psi:Z^*\rightarrow A^*$ st. 
$\tau=\phi^{-1}\cdot\iota_K\cdot\psi$. In this equation,  the relation $\iota_K\subseteq Z^*\times Z^*$ is defined by $\iota_K(z)=\{z\}\cap K$, for every $z\in Z^*$.
In addition we have $Z=A'\cup A''$, where $A'$ and $A''$ are disjoint copies of the alphabet $A$. 
From the point of view of implementation, some  finite automata with behaviors $\phi$, $\psi$, and $\iota_K$ can be explicitly computed. 
In particular, starting with  a finite $A^*\times A^*$-automaton with behavior $\tau$, say ${\cal A}$,
a finite automaton ${\cal A}'$ with behavior $\iota_k$ can be constructed by associating, with each transition $p\xrightarrow{(a_1\cdots a_m,b_1\cdots b_n)}q$  (with $m,n\ge 0$) in  ${\cal A}$,
the  transition $p\xrightarrow{a'_1\cdots a'_mb''_1\cdots b''_n}q$ in ${\cal A}'$.
The mapping $\phi$ (resp., $\psi$), for its part, is constructed on the basis of the  free monoid projection $\pi_{A'}:Z^*\rightarrow A'^*$ (resp.,  $\pi_{A''}:Z^*\rightarrow A''^*$).
By the way,  a finite $A^*$-automaton with behavior $\tau(X)$ can be effectively constructed. 
The reader could find more precise details in \cite[Sect. IV.1.3.1]{S03}).

\medbreak\noindent
(v) {\it Deciding whether a regular set is a  code}

Given  a regular set $X$, the question can be classically solved by applying Sardinas and Patterson algorithm \cite{SP53}.
Starting with the set $X^{-1}X\setminus\{\varepsilon\}$, an ultimately periodic sequence, say $(U_n)_{n\ge 0}$,
is computed by applying the following induction  formula: $U_{n+1}=U_n^{-1}X\cup X^{-1}U_n$.
In view of Prop. \ref{property-rec1}, each term of $(U_n)_{n\ge 0}$ is a regular set.
The algorithm necessarily stops: this corresponds to either  $\varepsilon\in U_n$, or  $U_n=U_p$, for some pair of different integers $p<n$: 
 $X$ is  a code iff.   the second cond. holds. 
\begin{example}
{\rm Over $A=\{a,b\}$, the set $X=\{a,ab,baa\}$ is not a code.
Indeed, we have $U_0=\{b\}$, $U_1=U_0^{-1}X\cup X^{-1}U_0=\{aa\}$, $U_2=U_1^{-1}X\cup X^{-1}U_1=\{a\}$, and $U_3=U_2^{-1}X\cup X^{-1}U_2=\{\varepsilon, b\}$.
We verify that the equation $ab\cdot a\cdot a    = a\cdot baa$ holds among the words of $X$.
}
\end{example}
\end{document}